\documentclass[journal]{IEEEtran}
\ifCLASSINFOpdf
  \usepackage[pdftex]{graphicx}
\else
  \usepackage[dvips]{graphicx}
\fi
\usepackage{amssymb}
\usepackage{amsmath}
\usepackage{amsthm}
\usepackage{caption}
\usepackage{subcaption}
\usepackage{booktabs}
\usepackage{algorithmic}

\newcommand{\trp}{\scriptscriptstyle\top}
\newtheorem{theorem}{Theorem}
\newtheorem{prop}{Proposition}
\newtheorem{corr}{Corollary}
\newtheorem{lemma}{Lemma}
\newtheorem{assumption}{Assumption}
\newtheorem{remark}{Remark}

\usepackage{accents} 

\begin{document}
%
\title{Scalar Field Estimation with Mobile Sensor Networks}
%
%
%

\author{Rihab~Abdul~Razak~,
        ~Srikant~Sukumar,~\IEEEmembership{Member,~IEEE,}
        and~Hoam~Chung,~\IEEEmembership{Member,~IEEE}
\thanks{R. A. Razak was with the IITB-Monash Research 
Academy, Mumbai, India.
e-mail: rihab@sc.iitb.ac.in.}
\thanks{S. Sukumar is with Systems and Control Engineering, 
Indian Institute of Technology Bombay, Mumbai, India.}
\thanks{H. Chung is with Mechanical and Aerospace 
Engineering, Monash University,
Clayton, VIC, Australia.}
}

%
%

\markboth{Journal of \LaTeX\ Class Files,~Vol.~14, No.~8, August~2015}%
{Shell \MakeLowercase{\textit{et al.}}: Bare Demo of IEEEtran.cls for IEEE Journals}
%



\maketitle

\begin{abstract}
In this paper, we consider the problem of estimating a scalar 
field using a network of mobile sensors which can measure the 
value of the field at their instantaneous location. The scalar 
field to be estimated is 
assumed to be represented by 
positive definite radial basis kernels and we use techniques 
from adaptive control and Lyapunov analysis to prove the 
stability of the proposed estimation algorithm. 
The convergence of the estimated parameter values to the true 
values is guaranteed by planning the motion of the mobile 
sensors to satisfy persistence-like conditions.
\end{abstract}


%
\IEEEpeerreviewmaketitle

\section{Introduction}
\label{sec:introduction}
%
%
%
%
\par Multi-robot systems consists of network of robots which 
cooperate to perform tasks such as consensus, formation control 
etc.
\cite{jadbabaie2003coordination,murray2007recent,
olfati2007consensus,tanner2007flocking}. 
Mobile sensor networks consists of network of robots mounted 
with sensors deployed to perform some distributed sensing task 
such as monitoring, coverage etc \cite{cortes2004coverage}.
In this paper, we consider the problem of estimation of an 
unknown scalar field using mobile sensor networks. There have 
been many works related to scalar field estimation in literature. 
Several works have studied field estimation using wireless sensor 
networks. See for example \cite{novak2004a,bajwa2005a}. 
In \cite{waterschoot2011a} the scalar field is assumed to be 
modelled using a partial differential equation and finite 
element methods are used for estimating the field. 
In \cite{vuran2006a,zhang2009a,dardari2007a,dogandzic2006a,
graham2012a,nevat2013a} the field is modelled as spatial 
random process and estimated using samples from the sensor nodes. 
In \cite{ramachandran2017a} field reconstruction is posed as  
an optimization problem constrained by linear dynamics and a 
gradient-based method is used to solve the problem.
In \cite{bergamo2012a}, the scalar field is assumed to be 
linearly parameterized in terms of Gaussian basis functions 
and the measurements from the sensors are fused together to 
form an estimate for the scalar field.

In most of these cases, the sensors are assumed to be fixed 
and distributed over the region of interest. Usually a large 
number of sensors are required to be installed for achieving  
enough spatial resolution. Using mobile sensor networks can be 
highly advantageous since they can move around the region of 
interest and collect measurements adaptively, the number of 
sensors required is greatly reduced. In \cite{la2013a,la2015a}, 
scalar field estimation is done with mobile sensor network by 
fusing sensor measurements using consensus filters. In 
\cite{wu2011a}, information about a scalar field is obtained 
by exploring the level surfaces of the field using a mobile 
sensor network. In \cite{zhang2007a}, a static sensor network 
is used along with a mobile robot to estimate a scalar field 
by combining the robot measurements with the sensor network 
measurements and planning the robot trajectory to minimize 
some reconstruction error. However the method we propose in 
the current work is motivated by the coverage control problem 
\cite{cortes2004coverage, Schwager2009, rihab2016a, 
rihab2018a, rihab2018b}.

\par In the coverage problem, we are interested in controlling 
the robots so that the robots attain an optimal configuration or 
a near optimal configuration with respect to a scalar field. 
In \cite{cortes2004coverage}, this is achieved by 
minimizing a cost function which gives a measure of how good the coverage is. In \cite{Schwager2009}, the authors extended the coverage algorithm for the case where the scalar field is unknown. The scalar field is assumed to be linearly parameterized with unknown constant parameters. In order to achieve the coverage goal, the robot needs to adapt the unknown parameters so that the estimated scalar field is close to the actual field. The exact estimation (asymptotically) of the density function parameters require a time integral quantity to be positive definite, which is a 
sufficient richness condition for the robot trajectories. See \cite{Schwager2009} for more details. 
In general, the robot trajectories need not meet this condition since the trajectories of the robots are decided based on the gradient of the coverage cost function, not on estimating the density function parameters. However, it is crucial to estimate the true values of those parameters since the estimation of the unknown scalar field is often the primary objective for a robotic sensor network and it may lead more efficient deployment of robots. 
For example, in case of radiation spill, if we have a good estimate of the radiation concentration, we may directly deploy agents to regions of high concentration.

Thus in this work, we look at a slightly different problem closely related to and motivated by the coverage problem discussed above. Our primary aim in this paper is to accurately estimate the scalar field not the coverage. The unknown scalar field is approximated using positive definite radial basis functions and we use a similar adaptive approach as that in 
\cite{Schwager2009} for parameter estimation. 

\par In Section \ref{sec:problem}, we discuss the problem 
statement in detail. In section \ref{sec:singleagent}, we 
consider the single mobile sensor case, followed by the mobile 
sensor network case in Section \ref{sec:multiagent}. In Section 
\ref{sec:centresnotknown} we discuss the case where the centres 
of the radial basis functions are not known exactly, but only 
to within an $\epsilon$-accuracy. We present some simulations 
to verify the results in section \ref{sec:simulations}. We 
conclude the paper with Section \ref{sec:conclusion}.

\section{Preliminaries and Problem Statement} \label{sec:problem}
We denote the set of positive real numbers by $\mathbb{R}_+$. 
The components of a vector $v$ are denoted using superscripts 
$v^i$. Subscripts on vector quantities refer to the agent or 
mobile sensor the quantity is associated to. For example, 
$v_i$ refers to a quantity associated with agent $i$.
\par We consider a compact region 
$\mathcal{Q} \subset \mathbb{R}^n$ with $N$ mobile 
sensors. The position of the sensors is denoted by 
$x_i; \,\,\, i=1,2,\dots,N$. 
There also exists a continuous scalar field 
$\phi : \mathcal{Q} \to \mathbb{R}_+$ 
over $\mathcal{Q}$ which is unknown. The objective  
is to estimate the unknown scalar field using $N$ mobile sensors 
assuming the sensors can measure the value of the scalar field 
at their respective positions. We assume that the unknown 
scalar field can be represented by positive definite 
radial basis functions (RBF). In other words, we assume the 
density function can be parameterized as 
\begin{subequations}
\begin{align}
	\label{eqn:phi1}
	\phi(q) &= \mathcal{K}(q)^{\trp} a \\
		&= \sum_{i=1}^p \mathcal{K}^{i}(q) a^{i}
\end{align}
\end{subequations}
where 
$a \in \mathbb{R}^p$ is a constant vector, and $\mathcal{K}(q) = 
\left[ \mathcal{K}^{1}(q) \,\,\, \mathcal{K}^{2}(q) \,\,\, 
\dots \,\,\, \mathcal{K}^{p}(q) \right]^{\trp}$ 
with $\mathcal{K}^i : \mathcal{Q} \to \mathbb{R}_+$ given by 
$\mathcal{K}^{i}(q) = \varphi(\|c_i - q\|)$ with 
$\varphi: \mathbb{R}_+ \to \mathbb{R}_+$ are radial basis 
functions for a set of points $c_i$. 
This assumption is common in neural 
networks and justified as follows: 
\begin{theorem}[\cite{park1991,lavretsky2013book}]
 \label{thm:approx1}
 For any continuous function $f:\mathbb{R}^n \to \mathbb{R}$ and 
 any $\epsilon > 0$, there is an RBF network with $p$ elements, 
 a set of centers $\{c_i\}_{i=1}^p$, 
 such that we can define 
 \begin{align*}
  \hat{f}(q) &= \sum_{i=1}^p a^{i} \mathcal{K}^{i}(q) 
  = a^{\trp} \mathcal{K}(q) 
 \end{align*}
 with $\|f-\hat{f}\|_{L_2}^2 \leq \epsilon = 
 \mathcal{O}\left(p^{-\frac{1}{n}}\right)$.
\end{theorem}
The theorem tells us that we can approximate a continuous 
function to an arbitrary accuracy by using a network of RBF 
elements. An example of positive definite radial kernel is the 
Gaussian kernel,  
\begin{equation}
	\mathcal{K}^{i}(q) = \varphi(\|c_i-q\|) =  
	\exp\left\{-\frac{\|c_i-q\|^2}{\sigma_i^2}\right\}
	\label{eqn:gaussian}
\end{equation} 
where $c_i$ are the centres of the Gaussian kernels.
The main problem studied in this work is to accurately determine 
the parameters $a^i$ so that the scalar field $\phi(.)$ may be 
accurately reconstructed. We make the following assumption:
\begin{assumption}
 The centres $c_i$ of the radial functions are known to all the 
 mobile agents.
\end{assumption}
The strengths $a_i$ of individual radial functions are unknown 
and need to be estimated. To proceed, we require the following 
theorem:
\begin{theorem}[Micchelli's Theorem \cite{micchelli1986}]
 \label{thm:approx2}
 Given $p$ distinct points $c_1,c_2,\dots,c_p$ in $\mathbb{R}^q$, 
 the $p \times p$ matrix $\mathrm{K}$, whose elements are 
 $\mathrm{K}_{ij} = \mathcal{K}^{i}(c_j) = \varphi(\|c_i-c_j\|)$ 
 is non-singular.
\end{theorem}
The theorem says that for positive definite radial kernels, 
the $p \times p$ matrix formed by evaluating the radial functions 
at each of the centres is non-singular.
In what follows, we assume that $\phi(.)$ can be 
\emph{exactly parameterized} by the RBF kernels. 
A consequence of theorem \ref{thm:approx2} is given below:

%
\begin{lemma}
 The matrix $S$ given by 
 \begin{equation}
  S:= \int_{\mathcal{Q}} \mathcal{K}(q) \mathcal{K}(q)^{\trp} dq
 \end{equation}
 where $\mathcal{K}(q) = \left[ \mathcal{K}^{1}(q) \,\,\, 
 \mathcal{K}^{2}(q) \,\,\, \dots \,\,\, \mathcal{K}^{p}(q) 
 \right]^{\trp}$ 
 and $\phi$ is parameterized as in (\ref{eqn:phi1}),
 is positive definite.
 \label{lemma:Spd}
\end{lemma}
\begin{proof}
 From the definition of $S$, we know it is atleast positive 
 semi-definite. 
 Therefore for any $v \neq 0$, $v^{\trp} S v \geq 0$ or 
 \[
 \int_{\mathcal{Q}} | \mathcal{K}(q)^{\trp}v |^2 dq \geq 0
 \]
 Now, since $\mathcal{K}(q)$ consists of positive definite 
 radial kernels, we have from theorem \ref{thm:approx2} that 
 \begin{equation*}
  \left( \begin{array}{cccc}
  \mathcal{K}^1(c_1) & \mathcal{K}^1(c_2) & \dots & 
  \mathcal{K}^1(c_p) \\
  \mathcal{K}^2(c_1) & \mathcal{K}^2(c_2) & \dots & 
  \mathcal{K}^2(c_p) \\
  \dots & \ddots & \dots & \vdots \\
  \mathcal{K}^p(c_1) & \mathcal{K}^p(c_2) & \dots & 
  \mathcal{K}^p(c_p) \\
  \end{array} \right)
 \end{equation*}
 is positive definite. This implies that the vectors 
 $\mathcal{K}(c_j); \,\, j=1,2,\dots,p$ 
 are linearly independent. Thus, given any 
 $v \neq 0, v \in \mathbb{R}^p$, 
 there exists some $j \in \{1,2,\dots,p\}$ such that 
 $\mathcal{K}(c_j)^{\trp} v$ is non-zero. 
 This along with the fact that $\mathcal{K}(\cdot)$ is 
 continuous allows us to conclude that
 \begin{equation*}
  \int_Q | \mathcal{K}(q)^{\trp} v |^2 dq > 0 
  \quad \mbox{for any } v \neq 0
 \end{equation*}
 Hence, $S$ is positive definite.
\end{proof}

\section{Single Mobile Robot Sensor} \label{sec:singleagent}
In this section, we consider the case of a single mobile sensor 
($N=1$) with position $x(t)$ at time $t$ deployed in the region 
$\mathcal{Q}$ to estimate the scalar field parameter $a$ 
(as given by equation \eqref{eqn:phi1}). The estimate of 
$a$ is denoted by $\hat{a}$. Then we can state the following 
corollary to lemma \ref{lemma:Spd}.
\begin{corr}
 Suppose the mobile sensor moves continuously within the domain 
 $Q$, such that in time $T$, it passes through each of the RBF 
 centres $c_i\, ; \,\, i=1,2,\dots,p$, then
 \begin{equation}
  \mathcal{S}_T := 
  \int_0^T \mathcal{K}(x(t)) 
  \mathcal{K}(x(t))^{\trp} dt
  \label{eqn:Tmatrix}
 \end{equation}
 is positive definite.
 \label{corr:pd}
\end{corr}
\begin{proof}
 The proof is essentially the same and follows from lemma 
 \ref{lemma:Spd}.
\end{proof}

\noindent Now consider the following integrators running on the 
mobile sensor:
\begin{equation}
 \begin{aligned}
  \dot{\Lambda} &= \mathcal{K}(t)\mathcal{K}(t)^{\trp} \\
  \dot{\lambda} &= \mathcal{K}(t)\phi(t) \\
 \end{aligned}
\end{equation}
where $\mathcal{K}(t) := \mathcal{K}(x(t))$ denotes the value 
of function $\mathcal{K}(\cdot)$ at the point where the robot is 
at time $t$ and $\phi(t)$ is the measured value of the density 
function $\phi(\cdot)$ by the robot at time $t$.

\begin{prop}
 \label{prop:1}
 Suppose the mobile sensor moves such that it passes through 
 each of the centres $c_i; \,\, i=1,2,\dots,p$ in some finite 
 time $T>0$, and during this motion updates its estimate 
 $\hat{a}$ of $a$ by 
 \begin{equation}
  \dot{\hat{a}} = -\Gamma \left( 
  \Lambda \hat{a} - \lambda \right),
  \label{eqn:single_updatelaw}
 \end{equation}
 where $\Gamma$ is a positive definite gain matrix, then the 
 estimate $\hat{a}$ is bounded and converges asymtotically to 
 the true value $a$.
\end{prop}
\begin{proof}
 Under the assumptions of the proposition \ref{prop:1} and 
 corollary \ref{corr:pd},
 \begin{equation*}
  S(T):= \int_0^T \mathcal{K}(\tau) 
  \mathcal{K}(\tau)^{\trp} d \tau
 \end{equation*}
 is positive definite. 
 This implies that 
 \begin{equation*}
  S(t) =
  \int_0^t 
  \mathcal{K}(\tau) \mathcal{K}(\tau)^{\trp} d \tau
 \end{equation*}
 is positive definite for all $t \geq T$. \\
 Now consider the positive definite candidate Lyapunov function,
 \begin{equation}
  V = \frac{1}{2} \tilde{a}^{\trp} \Gamma^{-1} \tilde{a}
 \end{equation}
 where $\tilde{a} = \hat{a} - a$ is the estimation error. 
 Taking the derivative of $V$, we obtain 
 \begin{equation*}
  \dot{V} = 
  \tilde{a}^{\trp} \Gamma^{-1} \dot{\hat{a}}
 \end{equation*}
 Substituting the update law from (\ref{eqn:single_updatelaw}) 
 and simplifying, we get
 \begin{equation*}
  \begin{aligned}
   \dot{V} &= -\tilde{a}^{\trp} S(t) \tilde{a} \\
   \dot{V}	&\leq \left\{ \begin{array}{ll}
   0 & \mbox{ for } t \in [0,T] \\
   -\alpha V & \mbox{ for } t > T,
   \end{array} \right.
  \end{aligned}
 \end{equation*}
 where $\alpha = \frac{\lambda_{\min}(S(T))}{\lambda_{\max}
 (\Gamma^{-1})} > 0$, $\lambda_{min}(\cdot)$ and 
 $\lambda_{\max}(\cdot)$ denoting the minimum and 
 maximum eigenvalues of their argument.
 Since $V$ is always non-increasing and bounded from below, 
 $\tilde{a}(t)$ is bounded for all $t > 0$.
 Since $\dot{V} < 0$ for all $t \geq T$, then we have 
 $V(t) \to 0$ as $t \to \infty$. This implies that 
 $\tilde{a} \to 0$ as $t \to \infty$.
\end{proof}
\begin{remark}
 The matrix $S(t)$ being positive definite for all $t \geq T$ is a 
 \emph{sufficient excitation} condition, similar to 
 (but weaker than) 
 the persistency of excitation condition, on the robot 
 trajectories which ensures parameter convergence. 
 See \cite{Schwager2009} for more information.
\end{remark}

\subsection{Relaxing the condition in corollary \ref{corr:pd}}
\label{sec:relaxpd}
In corollary \ref{corr:pd}, it was required that the mobile 
sensor passes through the centres $c_i$ of the radial kernels. 
This can be relaxed so that the mobile sensor need only move 
through a sufficiently small neighbourhood of each of the centres 
$c_i$, as described in \cite{gorinevsky1995persistency}. 
Consider the vector 
$\mathcal{X}(q) := \mathrm{K}^{-1} \mathcal{K}(q)$ where 
$\mathrm{K}$ is the matrix specified in theorem \ref{thm:approx2}. 
Then $\mathcal{X}(q)$ has the property that 
$\mathcal{X}^j(c_k) = \delta_{jk}$ where $\delta_{jk}$ is the 
Kronecker delta function and $\mathcal{X}^j(c_k)$ is the 
$j$-th component of $\mathcal{X}(c_k)$. Now consider the 
diagonal dominance sets defined by 
($0 < \varepsilon < 1$)
\[
 \mathcal{A}_j^{\varepsilon} := \left\{
 q \in \mathcal{Q} \, : \, |\mathcal{X}^j(q)| - 
 \sum_{i=1,i \neq j}^{p} |\mathcal{X}^i(q)| > \varepsilon
 \right\}.
\]
It can be easily seen that $\mathcal{A}_j^{\varepsilon}$ 
contains the centre $c_j$ and thus $\mathcal{A}_j^{\varepsilon}$ 
is an open subset containing $c_j$. The following lemma is an 
adaptation of theorem $1$ in \cite{gorinevsky1995persistency}:
\begin{lemma}
 Suppose that the mobile sensor moves continuously throughout 
 the domain $\mathcal{Q}$ such that in time $T$, the trajectory 
 traverses through each of the neighbourhoods 
 $\mathcal{A}_j^{\varepsilon} , \,\, j=1,2,\dots,p$, 
 then the matrix $\mathcal{S}_T$ given by equation 
 (\ref{eqn:Tmatrix}) is positive definite.
 \label{lemma:relaxedpd}
\end{lemma}
\begin{proof}
 $\mathcal{S}_T$ can be written as $\mathcal{S}_T = \mathrm{K} 
 \bar{\mathcal{S}}_T \mathrm{K}^{\trp}$ where 
 \[
  \bar{\mathcal{S}}_T = \int_0^T 
  \mathcal{X}(x(t)) \mathcal{X}(x(t))^{\trp} dt.
 \]
 Since $\mathrm{K}$ is invertible, $\mathcal{S}_T$ is positive 
 definite iff $\bar{\mathcal{S}}_T$ is positive definite. 
 $\bar{\mathcal{S}}_T$ is positive definite iff 
 there exists some $\delta > 0$ such that 
 $\underaccent{\bar}{\sigma}(\bar{\mathcal{S}}_T) \geq \delta$ 
 where $\underaccent{\bar}{\sigma}(A)$ denotes the minimum 
 singular value of $A$. Suppose $\bar{\mathcal{S}}_T$ is not 
 positive definite under the conditions of the theorem. 
 Then there exists no $\delta > 0$ such that 
 $\underaccent{\bar}{\sigma}(\bar{\mathcal{S}}_T) \geq \delta$. 
 This implies that for any $\delta > 0$, there exists 
 $u \neq 0, \|u\|=1 $ such that 
 $u^{\trp} \bar{\mathcal{S}}_T u < \delta$, i.e.,
 \[
  \int_0^T u^{\trp} \mathcal{X}(x(t)) 
  \mathcal{X}(x(t))^{\trp} u \,dt < \delta
 \]
 Let $i$ be the index of the components of $u$ which has the 
 largest absolute value. i.e., $|u^i| \geq |u^j| \,\, \forall j$. 
 Also let $[t_{i1},t_{i2}] \subset [0,T]$ be the subinterval 
 during which the mobile sensor trajectory is contained in the 
 set $\mathcal{A}_i^{\varepsilon}$. 
 Clearly since the set $\mathcal{A}_i^{\varepsilon}$ is open 
 and the trajectory is continuous, $[t_{i1},t_{i2}]$ has 
 finite positive length. Then,
 \begin{align}
  \int_0^T u^{\trp} & \mathcal{X}(x(t)) \mathcal{X}(x(t))^{\trp} 
  u \,dt	
  = \int_0^T |\mathcal{X}^{\trp}u|^2 \, dt \\
  & \geq \int_{t_{i1}}^{t_{i2}} |\mathcal{X}^{\trp}u|^2 \, dt
  = \int_{t_{i1}}^{t_{i2}} |\sum_{j=1}^p \mathcal{X}^j 
  u^j | ^2 \, dt \\
  & \geq \int_{t_{i1}}^{t_{i2}} ( |\mathcal{X}^i u^i| - 
  |\sum_{j=1,j \neq i}^p \mathcal{X}^j u^j|  )^2 \, dt \\
  & \geq \int_{t_{i1}}^{t_{i2}} ( |\mathcal{X}^i u^i| - 
  \sum_{j=1,j \neq i}^p |\mathcal{X}^j u^j|  )^2 \, dt \\
  & \geq \int_{t_{i1}}^{t_{i2}} (( |\mathcal{X}^i| - 
  \sum_{j=1,j \neq i}^p |\mathcal{X}^j| ) |u^i|  )^2 \, dt \\
  & \geq \int_{t_{i1}}^{t_{i2}} \varepsilon^2 |u^i|^2 \, dt 
  = (t_{i2} - t_{i1}) \varepsilon^2 |u^i|^2. 
 \end{align}
 Choosing $\delta < (t_{i2} - t_{i1}) \varepsilon^2 |u^i|^2$ 
 leads to a contradiction. Therefore, $\bar{\mathcal{S}}_T$ is 
 positive definite and hence $\mathcal{S}_T$ is positive 
 definite.
\end{proof}

\subsection*{A sufficient condition for satisfaction of lemma 
\ref{lemma:relaxedpd}'s assumptions:}
Since checking the condition of the mobile sensor traversing 
througn the sets $\mathcal{A}_j^{\varepsilon}$ in lemma 
\ref{lemma:relaxedpd} involves transforming the vector 
$\mathcal{K}(q)$ at each instant which can be cumbersome if 
the number of parameters are large, we present a simpler 
sufficient condition which ensures that a given point $q$ 
is inside the set $\mathcal{A}_j^{\varepsilon}$. Note that 
the conditions derived are not equivalent to the conditions of 
the lemma, but only sufficient and thus can be conservative. 
However it is beneficial during implementations.
\begin{lemma}
 Given the mobile sensor position $x$, if 
 \begin{equation}
  \|\mathcal{K}(x) - \mathcal{K}(c_j)\|_{\infty} < 
  \frac{(1 - \epsilon)}{2 (p-1) \|\mathrm{K}^{-1}\|_{\infty}},
 \end{equation}
 then $x \in \mathcal{A}_j^{\varepsilon}$.
\end{lemma}
\begin{proof} 
We have the $i$-th component of $\mathcal{X}(x)$, 
$\mathcal{X}^i(x) = \left[ \mathrm{K}^{-1} \mathcal{K}(x) 
\right]^i$. Then 
\begin{equation}
 \mathcal{X}^i(x) - \mathcal{X}^i(c_j) = \left[ \mathrm{K}^{-1} 
 (\mathcal{K}(x) - \mathcal{K}(c_j)) \right]^i
\end{equation}
Now consider the mapping
\begin{equation}
 \left[ \begin{array}{c}
 y_1 \\
 y_2
 \end{array} \right] = 
 B_j
 \left( \mathcal{X}(x) - \mathcal{X}(c_j) \right)
\end{equation}
where 
\begin{equation}
 B_j = 
 \left[ \begin{array}{ccccccc}
 0 & \dots & 0 & 1 & 0 & \dots & 0 \\
 1 & \dots & 1 & 0 & 1 & \dots & 1
 \end{array} \right]
\end{equation}
The $1$ in the first row and the $0$ in the second row occurs 
at the $j$-th column. If the infinity-norm of 
$y = [y_1, y_2]^{\trp}$, $\|y\|_{\infty}<(1-\varepsilon)/2$, 
then it is guaranteed that $x \in \mathcal{A}_j^{\varepsilon}$. 
We also have 
\begin{align}
 \|y\|_{\infty} & \leq \|B\|_{\infty} 
 \|\mathcal{X}(x) - \mathcal{X}(c_j)\|_{\infty} \\
 & \leq \|B\|_{\infty} \|\mathrm{K}^{-1}\|_{\infty} 
 \|\mathcal{K}(x) - \mathcal{K}(c_j)\|_{\infty}
\end{align}
Requiring the above bound to be less than 
$\frac{(1-\epsilon)}{2}$ and noting that 
$\|B\|_{\infty} = (p-1)$ we have 
\begin{equation}
 \|\mathcal{K}(x) - \mathcal{K}(c_j)\|_{\infty} < 
 \frac{(1 - \epsilon)}{2 (p-1) \|\mathrm{K}^{-1}\|_{\infty}}
\end{equation}
\end{proof}
Any point $p$ which satisfies the above condition will lie in 
the set $\mathcal{A}_j^{\varepsilon}$ although all points in 
$\mathcal{A}_j^{\varepsilon}$ are not characterized by the 
above condition.

\section{Mobile Sensor Network} \label{sec:multiagent}
Suppose that we have $N$ mobile sensors deployed in the region 
$\mathcal{Q}$, with the position of the $i$-th mobile sensor 
denoted by $x_i$. We want to estimate the function 
$\phi: \mathcal{Q} \to \mathbb{R}_+$ collectively. We assume 
that equation (\ref{eqn:phi1}) holds so that we can linearly 
parameterize $\phi(\cdot)$ in terms of radial basis functions. 
We partition the region into $N$ components 
$\mathcal{Q}_i \,\, (i=1,2,\dots,N)$. Correspondingly we 
partition the basis function vector $\mathcal{K}(q)$ and the 
parameter vector $a$ as 
\begin{equation}
 \mathcal{K}(q) = \left[ \begin{array}{c}
 \mathcal{K}^{(1)}(q) \\
 \mathcal{K}^{(2)}(q) \\
 \vdots \\
 \mathcal{K}^{(N)}(q) 
 \end{array} \right],
 \qquad 
 a = \left[ \begin{array}{c}
 a^{(1)} \\
 a^{(2)} \\
 \vdots \\
 a^{(N)} 
 \end{array} \right]
 \label{eqn:partition}
\end{equation}
Each region $\mathcal{Q}_i$ contains the centres of the basis 
functions in the sub-vector $\mathcal{K}^{(i)}$. We assign 
each region $\mathcal{Q}_i$ to one of the mobile sensors where 
the sensor operates. This assignment is permanent and 
each mobile sensor starts within its region $\mathcal{Q}_i$ 
and moves in $\mathcal{Q}_i$. The algorithms presented below 
do not depend on any particular partition or assignment of 
mobile sensors, and this can be done arbitrarily. One particular 
method to divide the region and assign the sensors will be 
discussed in section \ref{sec:simulations}. Assuming the region 
$\mathcal{Q}$ is partitioned and the mobile sensors are 
assigned to each partition, we consider the graph $\mathcal{G}$ 
with the vertices representing the mobile sensors and an edge 
existing between two sensors if they belong to adjacent 
partitions. By adjacent partitions, we mean two partitions which 
share a subset of their boundary with each other that is of 
non-zero length.
See figure \ref{fig:estimation_distributed} for an illustration.
\begin{figure}
 \centering
 \includegraphics[scale=1.8]{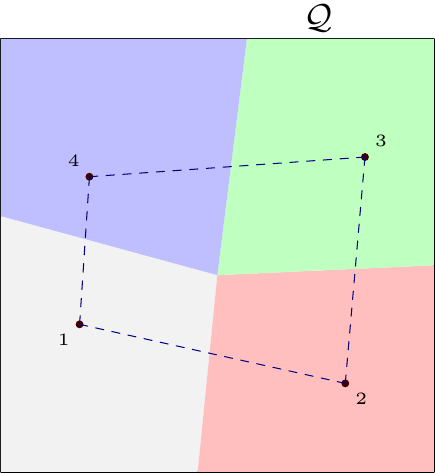}
 \caption{Illustration of four mobile sensors with a partition 
 of domain $\mathcal{Q}$: A graph with mobile sensors as root 
 nodes and edge between neighbouring sensors is also depicted 
 in the figure.}
 \label{fig:estimation_distributed}
\end{figure}
Now we consider two cases: (1) each mobile sensor estimates 
the entire parameter vector, and (2) each mobile sensor 
estimates only part of the parameter vector.

\subsection{Each mobile sensor estimates the full parameter vector}
\label{sec:ntwk_full}
\par \noindent In this subsection, we consider the case where 
each mobile sensor estimates the entire parameter vector, the 
estimate of sensor $i$ being denoted by $\hat{a}_i$.
To proceed, we consider the following integrators running on 
mobile sensor $i$:
\begin{align}
 \dot{\Lambda}_i &= 
 \mathcal{K}_i(t) \mathcal{K}_i(t)^{\trp} \\
 \dot{\lambda}_i &= \mathcal{K}_i(t) \phi_i(t)  
\end{align}
where $\mathcal{K}_i(t) = \mathcal{K}(x_i(t))$ and 
$\phi_i(t) = \phi(x_i(t))$ is the measurement of $\phi(.)$ 
obtained by sensor $i$ at its location at time $t$.
\par \noindent We consider the following update law for the 
parameter estimate of mobile sensor $i$:
\begin{equation}
 \dot{\hat{a}}_i = -\Gamma \left( 
 \Lambda_i \hat{a}_i - \lambda_i \right) 
 - \Gamma \zeta \sum_{j=1}^N l_{ij} \left( 
 \hat{a}_i - \hat{a}_j \right)
 \label{eqn:adaptationlaw_estimationmultiagent1}
\end{equation}
with $\hat{a}_i(0)$ being arbitrary; 
where $\zeta$ is a positive constant, $l_{ij}$ is the weight of 
the edge between sensors $i$ and $j$. The weight $l_{ij}$ is 
zero if there is no edge between sensor $i$ and $j$ and positive 
otherwise. The first term corresponds to the measurement 
update of mobile sensor $i$ and the second term is a consensus 
term to ensure that the estimates of all the mobile sensors 
asymptotically agree or come close to each other. 
This is critical in establishing the convergence of the 
estimation error as will be shown below.
\begin{lemma}
 \label{lemma:estimation_multiagent1}
 Suppose the mobile sensors translate continuously such that 
 in some time $T>0$, each sensor $i$ passes through each of the 
 centres in the region $\mathcal{Q}_i$ so that 
 \[
  \int\limits_{0}^T \mathcal{K}_i^{(i)}(t) 
  \mathcal{K}_i^{(i)}(t)^{\trp} dt > 0, 
  \quad \mbox{for} \,\, i=1,2,\dots,N.
 \]
 where $\mathcal{K}_i^{(i)}(t)$ denotes part of the vector 
 $\mathcal{K}_i(t)$ corresponding to the partition 
 \eqref{eqn:partition}.
 Then, we have 
 \[
  \sum_{i=1}^N 
  \int\limits_{0}^T \mathcal{K}_i(t) 
  \mathcal{K}_i(t)^{\trp} dt > 0.
 \]
\end{lemma}
\begin{proof}
Since each mobile sensor $i$ passes through the centres in 
the region $\mathcal{Q}_i$, the union of the trajectories of 
all mobile sensors cover all the centres, which implies that 
the matrix 
\begin{equation}
 \sum_{i=1}^N 
 \int\limits_{0}^T \mathcal{K}_i(t) 
 \mathcal{K}_i(t)^{\trp} dt
 \label{eqn:matrix-lambdai}
\end{equation}
is positive definite using the same arguments as in proof of  
corrollary \ref{corr:pd} and lemma \ref{lemma:Spd}.
\end{proof}
\begin{remark}
Lemma \ref{lemma:estimation_multiagent1} states that each agent 
passing through the centres in its partition $\mathcal{Q}_i$ 
is sufficient to ensure that the total sum matrix 
(\ref{eqn:matrix-lambdai}) is positive definite.
\end{remark}
\par \noindent Now we have the following result:
\begin{theorem}
\label{thm:estimation_multiagent1}
Suppose the $N$ mobile sensors adopt the parameter adaptation law 
(\ref{eqn:adaptationlaw_estimationmultiagent1}). Further assume 
that each mobile sensor $i$ traverses a trajectory going through 
all the basis function centres in $\mathcal{Q}_i$. Then 
\begin{equation}
 \lim_{t \to \infty} \left( \hat{a}_i - a \right) = 0,
\end{equation}
for each $i \in \{1,2,\dots,N\}$, i.e. the mobile sensors arrive 
at a common value for the parameters, the common value being the 
true parameter value.
\end{theorem}
\begin{proof}
Consider the function 
\begin{equation}
 V = \frac{1}{2} \sum_{i=1}^N 
 \tilde{a}_i^{\trp} \Gamma^{-1} \tilde{a}_i.
\end{equation}
Taking the derivative of $V$,
\begin{align*}
 \dot{V} &= \sum_{i=1}^N 
 \tilde{a}_i^{\trp} \Gamma^{-1} \dot{\hat{a}}_i \\
 &= -\sum_{i=1}^N \tilde{a}_i^{\trp} 
 \left( \Lambda_i \hat{a}_i - \lambda_i \right) - 
 \zeta \sum_{i=1}^N \tilde{a}_i^{\trp} l_{ij} 
 \left( \hat{a}_i - \hat{a}_j \right)
\end{align*}
Substituting for the variables $\Lambda_i$, $\lambda_i$ and 
rearranging the second term,
\begin{align}
 \dot{V} &= -\sum_{i=1}^N \tilde{a}_i^{\trp} 
 \int\limits_0^t 
 \mathcal{K}_i(\tau) \mathcal{K}_i^{\trp}(\tau) 
 d\tau \tilde{a}_i - 
 \zeta \sum_{\alpha=1}^p \hat{a}^{{\alpha}^{\trp}} 
 L \hat{a}^{\alpha}\\
 & \leq 0.
\end{align}
where 
$\hat{a}^{\alpha} = 
[ a_1^{\alpha} \, a_2^{\alpha} \, \dots \, a_N^{\alpha}]^{\trp}$ 
is the vector of the estimate of parameter $\alpha$ of all the 
sensors stacked together. The function $V$ is lower bounded and 
non-increasing, and therefore tends to a limit. This implies 
that $\dot{V}$ is integrable and also that the estimates 
$\hat{a}_i$ are bounded. $\dot{V}$ is also uniformly continuous 
since the derivative of each term in $\dot{V}$ is bounded. 
Using Barbalat's lemma, we conclude that $\dot{V}$ tends to zero 
as $t \to \infty$. From the second term in $\dot{V}$, noting 
that $L$ is the laplacian matrix of the connected graph 
$\mathcal{G}$ with nullspace $k\mathbf{1}$ where 
$\mathbf{1}$ is the vector of ones and $k \in \mathbb{R}$, 
we see that as $t \to \infty$, 
$\hat{a}^{\alpha} \to k_{\alpha} \mathbf{1}$ for some 
$k_{\alpha}$. Then,
\[
 \lim_{t \to \infty} \left( 
 \hat{a}_i - \hat{a}_j \right) = 0.
\]
since $\hat{a}_i = [ a_i^{1} \, a_i^{2} \, \dots \, 
a_i^{p}]^{\trp}$. Now from the first term of $\dot{V}$ we have, 
as $t \to \infty$, \[ - \tilde{a}^{\trp} \sum_{i=1}^N 
\int\limits_0^t 
\mathcal{K}_i(\tau) \mathcal{K}_i^{\trp}(\tau) 
d\tau \tilde{a} = 0 \]
where $\tilde{a}$ is the common value to which the mobile sensor 
parameter estimation errors $\tilde{a}_i$ converge.
Then using lemma \ref{lemma:estimation_multiagent1}, 
it follows that $\lim_{t \to \infty} \tilde{a} = 0$ and the 
parameter estimates converge to the true parameter values.
\end{proof}
\begin{remark}
Although lemma \ref{lemma:estimation_multiagent1} and theorem 
\ref{thm:estimation_multiagent1} requires that the mobile sensors 
move through the centres, the relaxation given in section 
\ref{sec:relaxpd} (requiring that the mobile sensors move only 
through the neighbourhoods $\mathcal{A}_j^{\varepsilon}$ of the 
centres) also applies here, as well as in all the following 
results which requires the sensors to move through the centres.
\end{remark}

\subsection{Each mobile sensor estimates only part of the 
parameter vector}
\label{sec:ntwk_part}
If the number of parameters $p$ is large as could be the case 
when the density function is completely unknown, each mobile 
sensor estimating the entire parameter vector could be 
computationally intensive, as it would require computing 
$\left( \frac{p(p+1)}{2} + p \right)$ filter variables in 
addition to the $p$ parameter estimates. In such cases it would 
be beneficial to have each mobile sensor estimate only part of 
the parameters. Suppose each mobile sensor $i$ is to estimate 
only part of the $a$-vector $a^{(i)}$ given by 
(\ref{eqn:partition}). Now we use $\hat{a}_i$ to 
denote the estimate of $a^{(i)}$ by sensor $i$. We write 
\begin{align} 
 \phi(q) &= \mathcal{K}(q)^{\trp} a \\
 &= \mathcal{K}^{(i)}(q)^{\trp} a^{(i)} + 
 \bar{\mathcal{K}}^{{(i)}^{\trp}} \bar{a}^{(i)}.
\end{align}
where $\mathcal{K}(q)$ and the parameter $a$ are partitioned 
appropriately. Since the mobile sensor $i$'s measurement is 
denoted by $\phi_i(t) := \phi(x_i(t))$, we have 
\begin{align}
 \phi_i(t) &= \mathcal{K}_i^{(i)}(t)^{\trp} a^{(i)} + 
 \bar{\mathcal{K}}_i^{(i)}(t)^{\trp} 
 \bar{a}^{(i)} \\
 &= \mathcal{K}_i^{(i)}(t)^{\trp} a^{(i)} 
 + \Delta \phi_i(t)
\end{align}
where $\mathcal{K}_i(t) := \mathcal{K}(x_i(t))$ and 
$\Delta \phi_i(t) := \bar{\mathcal{K}}_i^{(i)}(t)^{\trp} 
\bar{a}^{(i)}$. 
The basis functions in 
$\bar{\mathcal{K}}_i^{(i)}(t)$ are centred outside the region 
$\mathcal{Q}_i$ and thus their values at the points $p_i(t)$ 
are assumed to be small. Under this condition, we consider the 
contribution to $\phi(.)$ from these terms as a disturbance 
$\Delta \phi_i(t)$. 
\par \noindent Let $C = \{c_1,c_2,\dots,c_p\}$ be the set of 
centres of the basis functions, $C_i \subset C$ be its subset 
which belongs to $\mathcal{Q}_i$. We can then bound 
$\Delta \phi_i(t)$ as follows: 
\begin{lemma}
 For each mobile sensor $i$, $i \in \{1,2,\dots,N\}$,
 \begin{equation}
  |\Delta \phi_i(t)| \leq p \delta_i a_{\max}.
 \end{equation}
where $\delta_i := \max\limits_{j \in \{1,\dots,p\} } 
\exp\left\{-\frac{d_i^2}{\sigma_j^2}\right\}$, 
$d_i := \mbox{dist}(C_i,C \setminus C_i)$, 
$\mbox{dist}(A,B) = \min\limits_{a \in A, b \in B} \|a-b\|$, 
and $a_{\max}$ is an upper bound for the parameters, i.e., 
$|a^i| \leq a_{\max} \,\, \forall i \in \{1,2,\dots,p\}$. \\
Further the bound can be made independent of $i$ as follows, 
\begin{equation}
 |\Delta \phi_i(t)| \leq p \delta a_{\max}.
\end{equation}
where $\delta = \max\limits_{j \in \{1,\dots,N\} } \delta_i$.
\end{lemma}
\begin{proof}
The lemma follows from the definition of $\Delta \phi_i(t)$ using 
Cauchy-Schwartz inequality.
\end{proof}
\par \noindent We again define the following integrators:
\begin{align}
 \dot{\Lambda}_i &= 
 s \mathcal{K}_i^{(i)} \mathcal{K}_i^{{(i)}^{\trp}} \\
 \dot{\lambda}_i &= s \mathcal{K}_i^{(i)} \phi_i
\end{align}
where $s$ is a switching signal which takes values in the set 
$\{0,1\}$.
Consider the following adaptation law:
\begin{equation}
 \dot{\hat{a}}_i = 
 -\Gamma \left( \Lambda_i \hat{a}_i - \lambda_i \right)
 \label{eqn:adaptationlaw_estimationmultiagent2}
\end{equation}
Then we have the following result:

\begin{theorem}
 \label{thm:estimation_multiagent2}
 Suppose the $N$ mobile sensors implement the parameter 
 adaptation law (\ref{eqn:adaptationlaw_estimationmultiagent2}) 
 with each sensor $i$ only estimating part of the full parameter 
 vector $a^{(i)}$. Further assume that each mobile sensor 
 $i$ produces a trajectory going through all the basis function 
 centres in $\mathcal{Q}_i$ in time $T>0$. Then 
 \[
  \lim_{t \to \infty} \|\hat{a}_i(t) - a^{(i)}\| \leq r_i,
 \]
 where $r_i = \frac{Tp \delta_i a_{\max}}{\alpha \eta_i}$, 
 $a_{\max}$ is the upper bound on the parameter values in 
 $a^{(i)}$, $\alpha \in (0,1)$ and $\eta_i$ is the smallest 
 eigen-vlaue of the matrix 
 $\int_0^T \mathcal{K}_i^{(i)} \mathcal{K}_i^{{(i)}^{\trp}} 
 d\tau $.
\end{theorem}
\begin{proof}
Consider 
\begin{equation}
 V = \frac{1}{2} \sum_{i=1}^N 
 \tilde{a}_i^{\trp} \Gamma^{-1} \tilde{a}_i
\end{equation}
Taking derivative,
\begin{align}
 \dot{V} &= - \sum_{i=1}^N \tilde{a}_i^{\trp} 
 \left( \Lambda_i \hat{a}_i - \lambda_i \right) \\
 &= - \sum_{i=1}^N \tilde{a}_i^{\trp} 
 \int\limits_0^t s \mathcal{K}_i^{(i)} \left(
 {\mathcal{K}_i^{(i)}}^{\trp} \hat{a}_i - 
 {\mathcal{K}_i^{(i)}}^{\trp} a^{(i)} 
 - \Delta \phi_i \right) d \tau \\
 &= - \sum_{i=1}^N \tilde{a}_i^{\trp} 
 \int\limits_0^t s \mathcal{K}_i^{(i)} \left(
 {\mathcal{K}_i^{(i)}}^{\trp}\tilde{a}_i 
 - \Delta \phi_i 
 \right) d \tau \\
 &= - \sum_{i=1}^N \tilde{a}_i^{\trp} 
 \int\limits_0^t s \mathcal{K}_i^{(i)} 
 {\mathcal{K}_i^{(i)}}^{\trp} 
 d \tau \, \tilde{a}_i 
 + \sum_{i=1}^N \tilde{a}_i^{\trp} 
 \int\limits_0^t s \mathcal{K}_i^{(i)} 
 \Delta \phi_i d \tau
\end{align}
For $t \geq T$, the first term becomes negative definite 
(assuming $s>0$). Setting $s=1$ for $t \leq T$ and $s=0$ for 
$t > T$, we have 
\begin{equation}
 \dot{V} = 
 - \sum_{i=1}^N \tilde{a}_i^{\trp} 
 \int\limits_0^T \mathcal{K}_i^{(i)} 
 {\mathcal{K}_i^{(i)}}^{\trp} 
 d \tau \, \tilde{a}_i 
 + \sum_{i=1}^N \tilde{a}_i^{\trp} 
 \int\limits_0^T \mathcal{K}_i^{(i)} 
 \Delta \phi_i d \tau
\end{equation}
for $t>T$.
Then
\begin{align}
 \dot{V} &\leq -\sum_{i=1}^N \eta_i \|\tilde{a}_i\|^2 
 + \sum_{i=1}^N \|\tilde{a}_i\| Tp \delta_i a_{\max} \\
 &\leq -\kappa V - \sum_{i=1}^N \|\tilde{a}_i\|
 \left(\alpha \eta_i \|\tilde{a}_i\| - Tp \delta_i a_{\max}\right)
\end{align}
where $\kappa = 
\frac{\eta_{\min}}{\lambda_{\max}(\Gamma^{-1})}$ and 
$\alpha \in (0,1)$. 
Thus for $\|\tilde{a}_i\| > r_i$, 
we have $\dot{V} \leq -\kappa V$ and $V$ decays exponentially. 
Therefore the statement of the theorem holds.
\end{proof}

\subsection{Improving the steady state error}
\label{sec:improve_steadystateerror}
In this section, we propose a strategy to improve the steady 
state error with the strategy in theorem 
\ref{thm:estimation_multiagent2}. Note that the strategy in 
theorem \ref{thm:estimation_multiagent2} is completely 
decentralized in that there is no real-time communication 
required between the mobile sensors to implement the estimation 
strategy. On the other hand, we can get better 
parameter estimates at the cost of exchanging information 
about parameter estimates with other mobile sensors.
\par The term $\Delta \phi_i(t)$ depends on the true value of 
parameters corresponding to the other mobile sensors 
(denoted $\bar{a}^{(i)}$). Since we do not know the true values, 
we cannot cancel this term and treat it as a disturbance. 
However we know that the other mobile sensors 
have estimates for the true values of $\bar{a}^{(i)}$. 
We can use these parameter estimates to reduce the effect of 
the $\Delta \phi_i(t)$ term on the estimation algorithm. 
Note that the vector $\bar{a}^{(i)}$ consists of the 
sub-vectors $a^{(j)}$ for all $j \neq i$. Now, corresponding 
to each $a^{(i)}$, we construct a directed graph with a rooted 
outbranching (see \cite{mesbahi2010graph}), denoted 
$\mathcal{G}_i$ which is a sub-graph of the undirected graph 
$\mathcal{G}$ with mobile sensor $i$ as the root node. 
An illustration is shown in figure \ref{fig:partition2}.
\begin{figure}
 \centering
 \includegraphics[scale=1.8]{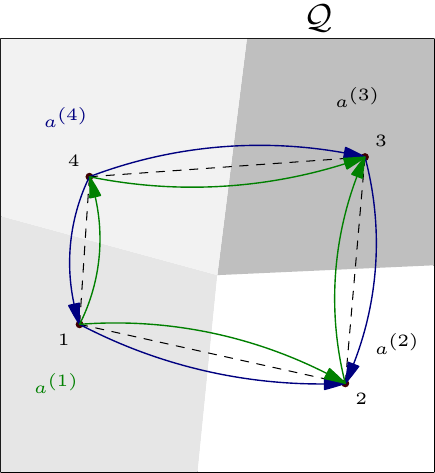}
 \caption{Illustration of four mobile sensors with the directed 
 graphs corresponding to $a^{(1)}$ and $a^{(4)}$.}
 \label{fig:partition2}
\end{figure}
\par For each mobile sensor $i$, we introduce additional states 
$b_i^j$ for each $j \in \{1,2,\dots,N\}$ and $j \neq i$, 
which evolves according to the equation  
\begin{equation}
 \dot{b}_i^j = - \sum_{k=1}^N l^d_{ik} 
 \left( \hat{b}_i^j - \hat{b}_k^j \right)
\end{equation}
where we define $b_i^i := \hat{a}_i$ for ease of notation and 
$l^d_{ik}$ is zero if there is no directed path from node $i$ 
to $k$ in graph $\mathcal{G}_j$, and non-zero constant value 
otherwise. This implements a directed consensus protocol on 
the variables $b_i^j$ with $i=1,2,\dots,N$ 
(see \cite{mesbahi2010graph}) converging to the root value 
$b_j^j = \hat{a}_j$ for each $j$. Thus $b_i^j$ is an estimator 
of $\hat{a}_j$ by mobile sensor $i$. We now use the modified 
integrators:
\begin{align}
 \dot{\Lambda}_i &= 
 s \mathcal{K}_i^{(i)} \mathcal{K}_i^{{(i)}^{\trp}} \\
 \dot{\lambda}_i &= s \mathcal{K}_i^{(i)} \left( \phi_i - 
 \bar{\mathcal{K}}_i^{{(i)}^{\trp}} b_i \right)
\end{align}
where $b_i$ is the concatenated vector given by $b_i = 
\left[ b_i^{1^{\trp}} \dots b_i^{j^{\trp}} \dots 
b_i^{N^{\trp}} \right]^{\trp}$ ($j=i$ not included). 
Using the adaptation law 
(\ref{eqn:adaptationlaw_estimationmultiagent2}) we can see that 
the disturbance term now becomes 
\begin{equation}
 \Delta \phi_i'(t) := 
 \bar{\mathcal{K}}_i^{(i)}(t)^{\trp} (\bar{a}^{(i)}-b_i)
\end{equation}
which is expected to be smaller than $\Delta \phi_i(t)$, 
although we cannot put a theoretical bound better than $r_i$ 
in theorem \ref{thm:estimation_multiagent2}. 
The stability and convergence in case of the above modification 
is not proved here as it is essentially a similar exercise to 
that in the previous section. We will investigate the effect 
of the above modification in section \ref{sec:simulations}.

\section{Unknown Centres} \label{sec:centresnotknown}
In this section, we assume as before that the scalar field is 
a finite linear combination of radial basis functions. 
We further assume that the centres are not exactly known, 
but known to within an accuracy of $\epsilon_c$, i.e., 
$\|\hat{c}_i - c_i\| \leq \epsilon_c$. We will evaluate 
the quality of parameter estimates in this case. Define 
\[
 \tilde{\mathcal{K}}(q) = \hat{\mathcal{K}}(q) - \mathcal{K}(q)
\]
where $\hat{\mathcal{K}}(q)$ is the RBF evaluated at the 
known values of the centres and $\mathcal{K}(q)$ corresponds 
to the true values of the centres.

\subsection{Each mobile sensor estimates only a part of 
the parameter vector}
As in section \ref{sec:ntwk_full}, we assume that each 
mobile sensor estimates part of the parameter vector 
$a^{(i)}$ corresponding to the partition $\mathcal{Q}_i$. 
In this case we propose the following modified filters,
\begin{align}
 \label{eqn:unknowncentres1_filters1}
 \dot{\Lambda}_i &= s\hat{\mathcal{K}}_i^{(i)} 
 \hat{\mathcal{K}}_i^{{(i)}^{\trp}} \\
 \dot{\lambda}_i &= s\hat{\mathcal{K}}_i^{(i)} \phi_i 
 \label{eqn:unknowncentres1_filters2}
\end{align}
with equation (\ref{eqn:adaptationlaw_estimationmultiagent2}) 
as the adaptation law. Then we have the following result.
\begin{prop}
 Assuming the centres are only known to within an accuracy of 
 $\epsilon_c$ ($\|\hat{c}_i - c_i < \epsilon_c\|$), 
 let each mobile sensor pass through the set of 
 known (inaccurate) centres $\hat{c}_i$ in $\mathcal{Q}_i$. 
 If each mobile sensor implements the adaptation law 
 \eqref{eqn:adaptationlaw_estimationmultiagent2} along with 
 \eqref{eqn:unknowncentres1_filters1}-
 \eqref{eqn:unknowncentres1_filters2}, 
 the estimation error $\tilde{a}_i$ converges to within a bound 
 $r_i$ of the origin, where $r_i = \frac{Tp a_{\max} 
 (\sqrt{p} k \epsilon_c + \delta_i)}{\alpha \eta_i}$.
\end{prop}
\begin{proof}
Consider the same Lyapunov function as before,
\[
 V = \sum_{i=1}^N \tilde{a}_i^{\trp} \Gamma^{-1} \tilde{a}_i
\]
Taking the time derivative,
\begin{align*}
 \dot{V} &= - \sum_{i=1}^N \tilde{a}_i^{\trp} 
 \left( \Lambda_i \hat{a}_i - \lambda_i \right) \\
 &= - \sum_{i=1}^N \tilde{a}_i^{\trp} 
 \int\limits_0^t s \hat{\mathcal{K}}_i^{(i)} \left(
 {\hat{\mathcal{K}}_i^{(i)}}^{\trp} \hat{a}_i - 
 {\mathcal{K}_i^{(i)}}^{\trp} a^{(i)} 
 - \Delta \phi_i \right) d \tau \\
 &= - \sum_{i=1}^N \tilde{a}_i^{\trp} 
 \int\limits_0^t s \hat{\mathcal{K}}_i^{(i)} 
 {\hat{\mathcal{K}}_i^{(i)}}^{\trp} 
 d \tau \, \tilde{a}_i 
 - \sum_{i=1}^N \tilde{a}_i^{\trp}
 \int\limits_0^t s \hat{\mathcal{K}}_i^{(i)} 
 {\tilde{\mathcal{K}}_i^{(i)}}^{\trp} 
 d \tau \, a^{(i)} \\
 & \qquad + \sum_{i=1}^N \tilde{a}_i^{\trp} 
 \int\limits_0^t s \hat{\mathcal{K}}_i^{(i)} 
 \Delta \phi_i d \tau
\end{align*}
Also note that $|\hat{\mathcal{K}}^i(q)| \leq 1 
\implies \|\hat{\mathcal{K}}(q)\| 
\leq \sqrt{p}$, and $|\tilde{\mathcal{K}}^i(q)| \leq k 
\epsilon_c \implies \|\tilde{\mathcal{K}}(q)\| \leq 
\sqrt{p}k \epsilon_c$ for some $k$ (lipschitz constant), 
Setting $s=1$ for $t \leq T$ and $s=0$ for $t > T$ as before 
and, assuming the first term becomes negative definite at 
time $T$, 
we now have 
\begin{align*}
 \dot{V} & \leq  
 - \sum_{i=1}^N \tilde{a}_i^{\trp} 
 \int\limits_0^T \hat{\mathcal{K}}_i^{(i)} 
 {\hat{\mathcal{K}}_i^{(i)}}^{\trp} 
 d \tau \, \tilde{a}_i \\
 & \qquad + \sum_{i=1}^N \|\tilde{a}_i\| Tp a_{\max} 
 (\sqrt{p}k \epsilon_c + \delta_i) \\
 & \leq
 - \kappa V - \sum_{i=1}^N \|\tilde{a}_i\|
 \left(\alpha \eta_i \|\tilde{a}_i\| - Tp a_{\max} 
 (\sqrt{p}k \epsilon_c + \delta_i) \right)
\end{align*}
for $t \geq T$.
Therefore, the statement of the theorem follows.
\end{proof}

\subsection{Each mobile sensor estimates the entire parameter 
vector}
We define the following filter equations, 
\begin{align}
 \dot{\Lambda}_i &= s\hat{\mathcal{K}}_i 
 \hat{\mathcal{K}}_i^{\trp} 
 \label{eqn:unknowncentres2_filters1} \\
 \dot{\lambda}_i &= s\hat{\mathcal{K}}_i \phi_i 
 \label{eqn:unknowncentres2_filters2}
\end{align}
The adaptation law is given by equation 
(\ref{eqn:adaptationlaw_estimationmultiagent1}).
In this case, we have the following proposition.
\begin{prop}
 Suppose the $N$ mobile sensors adopt the parameter adaptation law 
 (\ref{eqn:adaptationlaw_estimationmultiagent1}) with the 
 integrators \eqref{eqn:unknowncentres2_filters1}-
 \eqref{eqn:unknowncentres2_filters2}. Also assume 
 that each mobile sensor $i$ produces a trajectory going through 
 all the approximate basis function centres $\hat{c}_i$ in 
 $\mathcal{Q}_i$. Then the parameter estimation errors of the 
 mobile sensors converge to within a bound $r_i$ of origin, 
 where $r_i = \frac{Tp \sqrt{p} k \epsilon_c a_{\max}}
 {\alpha \eta_{\min}}$.
\end{prop}
\begin{proof}
Consider the lyapunov function 
\[
 V = \sum_{i=1}^N \tilde{a}_i^{\trp} \Gamma^{-1} \tilde{a}_i
\]
Taking the derivative of $V$,
\begin{align*}
 \dot{V} &= \sum_{i=1}^N 
 \tilde{a}_i^{\trp} \Gamma^{-1} \dot{\hat{a}}_i \\
 &= -\sum_{i=1}^N \tilde{a}_i^{\trp} 
 \left( \Lambda_i \hat{a}_i - \lambda_i \right) - 
 \zeta \sum_{i=1}^N \tilde{a}_i^{\trp} l_{ij} 
 \left( \hat{a}_i - \hat{a}_j \right)
\end{align*}
Substituting for the variables $\Lambda_i$, $\lambda_i$ and 
rearranging the second term,
\begin{align*}
 \dot{V} &= -\sum_{i=1}^N \tilde{a}_i^{\trp} 
 \int\limits_0^t 
 s \hat{\mathcal{K}}_i \hat{\mathcal{K}}_i^{\trp} 
 d\tau \tilde{a}_i 
 - \sum_{i=1}^N \tilde{a}_i^{\trp}
 \int\limits_0^t s \hat{\mathcal{K}}_i 
 {\tilde{\mathcal{K}}_i}^{\trp} 
 d \tau \, a^{(i)} \\
 & \qquad - \zeta \sum_{\alpha=1}^p \hat{a}^{{\alpha}^{\trp}} 
 L \hat{a}^{\alpha} 
\end{align*}
Simplifying,
\begin{align*}
 \dot{V} &= -\sum_{i=1}^N \tilde{a}_i^{\trp} 
 \int\limits_0^T 
 \hat{\mathcal{K}}_i \hat{\mathcal{K}}_i^{\trp} 
 d\tau \tilde{a}_i 
 - \sum_{i=1}^N \tilde{a}_i^{\trp}
 \int\limits_0^T \hat{\mathcal{K}}_i 
 {\tilde{\mathcal{K}}_i}^{\trp} 
 d \tau \, a^{(i)} \\
 & \qquad - \zeta \sum_{\alpha=1}^p \tilde{a}^{{\alpha}^{\trp}} 
 L \tilde{a}^{\alpha} \\
\end{align*}
for $t \geq T$.
We can write the first and last terms in the above equation 
in terms of stacked vectors as 
\begin{align*}
 \dot{V} &= - \tilde{\underline{a}}^{\trp} \underline{Q} 
 \tilde{\underline{a}} 
 - \zeta \, \tilde{\underline{a}}^{\trp} P^{\trp} 
 \underline{L} P \tilde{\underline{a}} 
 - \tilde{\underline{a}}^{\trp}
 E \underline{a} \\
 &= - \tilde{\underline{a}}^{\trp} \left( \underline{Q} 
 + \zeta \,  P^{\trp} 
 \underline{L} P \right) \tilde{\underline{a}} 
 - \tilde{\underline{a}}^{\trp}
 E \underline{a} 
\end{align*}
where $\tilde{\underline{a}} = \left[ \tilde{a}_1^{\trp} \, 
\tilde{a}_2^{\trp} \, 
\dots \, \tilde{a}_N^{\trp} \right]^{\trp}$, 
\begin{align*}
 \underline{Q} &= 
 \left[ \begin{array}{cccc}
 \int_0^T  \hat{\mathcal{K}}_1 \hat{\mathcal{K}}_1^{\trp} d\tau 
 & \dots & 0 \\
 0 & \dots & 0 \\
 \vdots & \ddots & \vdots \\
 0 & \dots & \int_0^T  \hat{\mathcal{K}}_N 
 \hat{\mathcal{K}}_N^{\trp} d\tau
 \end{array} \right], \\
 \underline{L} &= 
 \left[ \begin{array}{cccc}
 L & 0 & \dots & 0 \\
 0 & L & \dots & 0\\
 \vdots & \vdots & \ddots & \vdots \\
 0 & 0 & \dots & L
 \end{array} \right], \\
 E &= 
 \left[ \begin{array}{cccc}
 \int_0^T  \hat{\mathcal{K}}_1 \tilde{\mathcal{K}}_1^{\trp} d\tau 
 & \dots & 0 \\
 0 & \dots & 0 \\
 \vdots & \ddots & \vdots \\
 0 & \dots & \int_0^T  \hat{\mathcal{K}}_N 
 \tilde{\mathcal{K}}_N^{\trp} d\tau
 \end{array} \right]
\end{align*}
and $P$ is the permutation matrix 
\[
 P = \left[ \begin{array}{ccccccc}
 1 & 0 & \dots & 0 & 0 & \dots & 0 \\
 0 & 0 & \dots & 1 & 0 & \dots & 0 \\
 \vdots & \vdots & \vdots & \vdots & \vdots & \vdots & \vdots \\
 0 & 1 & \dots & 0 & 0 & \dots & 0 \\
 0 & 0 & \dots & 0 & 1 & \dots & 0 \\
 \vdots & \vdots & \vdots & \vdots & \vdots & \vdots & \vdots \\
 \end{array} \right]
\]
of dimension $Np \times Np$. We show that the matrix 
$ \left( \underline{Q} + P^{\trp} \underline{L} P \right) $
is positive definite. 
Each of the terms are positive semi-definite. 
The nullspace of matrix $\underline{L}$ contains elements 
of the form 
\[ c_1 \left[ \begin{array}{c} \mathbf{1}_p \\ 0 \\ \vdots \\ 0 
 \end{array} \right] + c_2 \left[ \begin{array}{c} 0 \\ 
 \mathbf{1}_p \\ \vdots \\ 0 \end{array} \right] + \dots + 
 c_N \left[ \begin{array}{c} 0 \\ 0 \\ \vdots \\ \mathbf{1}_p 
 \end{array} \right]. \] 
Therefore $P^{\trp} \underline{L} P$ has nullspace elements of 
the form \[ c_1 \left[ \begin{array}{c} 1 \\ 0 \\ \vdots \\ 0 \\ 
1 \\ 0 \\ \vdots \\ 0 \end{array} \right] + c_2 \left[ 
\begin{array}{c} 0 \\ 1 \\ \vdots \\ 0 \\ 0 \\ 1 \\ 
\vdots \\ 0 \end{array} \right] + \dots + c_N \left[ 
\begin{array}{c} 0 \\ 0 \\ \vdots \\ 1 \\ 0 \\ 0 \\ \vdots \\ 1 
\end{array} \right], \] i.e., elements of the form 
$\left[ \, c_1 \, c_2 \, \dots \, c_N \, c_1 \, c_2 \, \dots \, 
c_N \right]^{\trp}$. Correspondingly the $\underline{Q}$ term 
can be written as 
\[
 c^{\trp} \sum_{i=1}^N \int_0^T \hat{\mathcal{K}}_i 
 \hat{\mathcal{K}}_i^{\trp} dt \, c
\]
where $c = \left[ \, c_1 \, c_2 \, \dots \, c_N \right]^{\trp}$.
Under the assumptions of the proposition, and lemma 
\ref{lemma:estimation_multiagent1}, the above term is strictly 
positive. Hence 
$\left( \underline{Q} + P^{\trp} \underline{L} P \right)$ 
is positive definite. Let $\eta_{\min}$ be the smallest 
eigen-value of 
$\left( \underline{Q} + P^{\trp} \underline{L} P \right)$. 
Then we have 
\begin{align*}
 \dot{V} & \leq -\kappa V - \alpha \eta_{\min} 
 \|\underline{\tilde{a}}\|^2 
 + \sum_{i=1}^N \|\tilde{a}_i\| Tp \sqrt{p} k \epsilon_c a_{\max} 
 \\
 & = -\kappa V - \alpha \eta_{\min} \sum_{i=1}^N \|\tilde{a}_i\| 
 \left( \|\tilde{a}_i\| - \frac{Tp \sqrt{p} k \epsilon_c a_{\max}}
 {\alpha \eta_{\min}} \right)
\end{align*}
for some $\kappa > 0$.
Thus for 
$\|\tilde{a}_i\| > \frac{Tp \sqrt{p} k \epsilon_c a_{\max}}
{\alpha \eta_{\min}}$, $V$ decreases exponentially and the 
result holds.
\end{proof}

\section{Simulations} \label{sec:simulations}
In this section, we verify the algorithms presented using 
simulations. First we consider the exact parameterization case 
where the true scalar field is a linear combination of RBFs 
with the centres of the RBFs being known. This case allows us to 
verify the correctness of the algorithms presented in the paper. 
Next we consider a scalar field which is completely unknown, 
and use the algorithms presented to reconstruct the scalar field. 
The mobile sensors in the simulations are assumed to be single 
integrators with dynamics given by $\dot{x}_i = u_i$ where 
$x_i$ is the position of sensor $i$ and $u_i$ is its control 
input. For ease of comparing various algorithms, we refer to 
the algorithm in section \ref{sec:ntwk_full} as 
\emph{Algorithm S1}, the algorithm presented in section 
\ref{sec:ntwk_part} as \emph{Algorithm S2}, and the modified 
version of algorithm $S2$ in section 
\ref{sec:improve_steadystateerror} as \emph{Algorithm S3}.
\subsection{Exact parameterization}
We consider the unit square region $\mathcal{Q}$ with four 
mobile sensors. The scalar field to be estimated is exactly 
parameterized in terms of Gaussian RBFs 
(given by equation \eqref{eqn:gaussian}), the $x$ and $y$ 
coordinates of the RBF centres $c_i$ being given in table 
\ref{tab:sim1}. The standard deviation of each of the gaussians 
$\sigma_i$ is chosen to be $0.1$. The true parameter values $a^i$ 
are also given in table \ref{tab:sim1}.
\begin{table}
 \begin{tabular}{l|cccccccc}
  $c_{i,x}$ & $0.20$ & $0.35$ & $0.60$ & $0.85$ & $0.70$ & $0.75$ 
  & $0.15$ & $0.35$ \\	
  $c_{i,y}$ & $0.25$ & $0.26$ & $0.18$ & $0.30$ & $0.75$ & $0.90$ 
  & $0.75$ & $0.60$ \\	
  $a^{i}$ & $2.0$ & $1.0$ & $1.5$ & $1.8$ & $1.2$ & $1.6$ 
  & $2.5$ & $1.1$ \\	
 \end{tabular}
 \caption{Parameters of the simulated scalar field}
 \label{tab:sim1}
\end{table}
The scalar field is shown in figure \ref{fig:scalarfield1}.
\begin{figure}
 \begin{subfigure}{0.23\textwidth}
  \centering
  \includegraphics[scale=0.38]
  {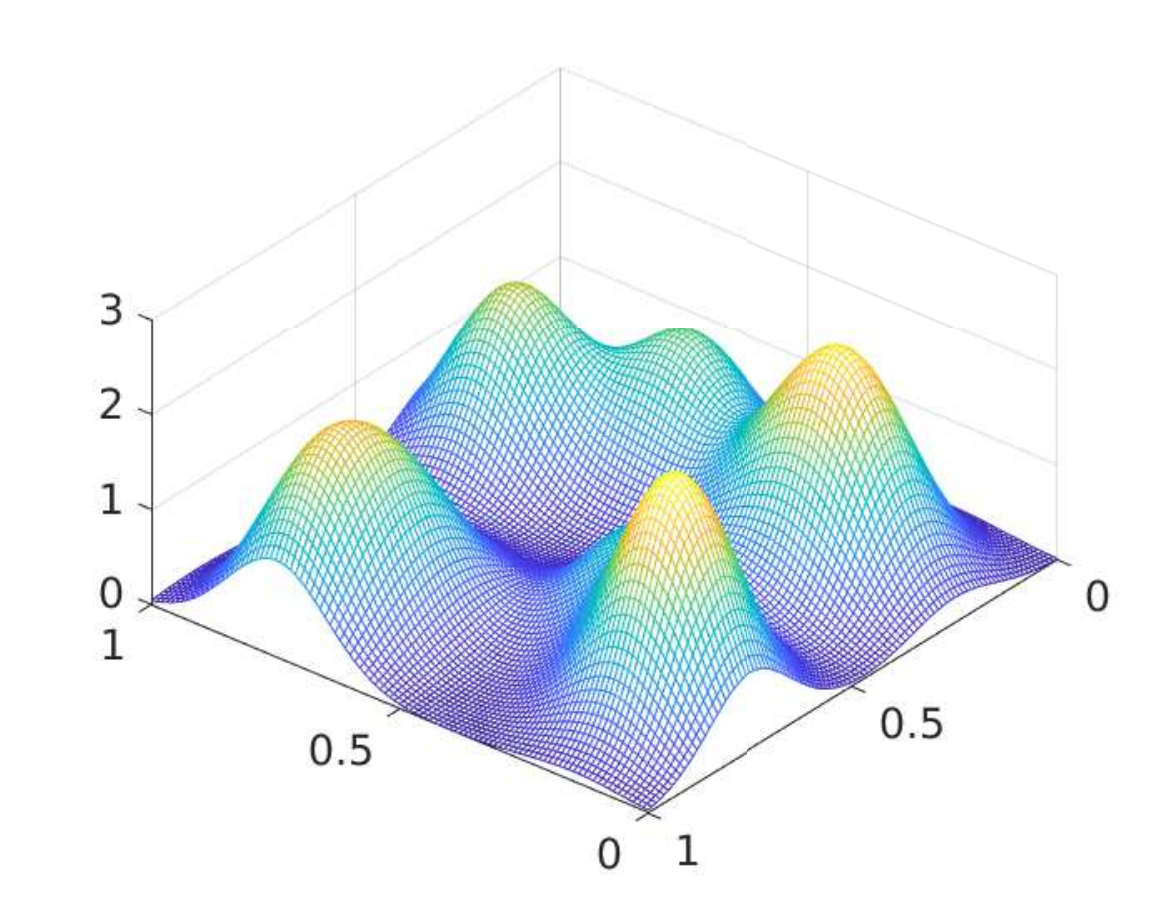}
 \end{subfigure}
 \begin{subfigure}{0.24\textwidth}
  \centering
  \includegraphics
  [scale=0.45]
  {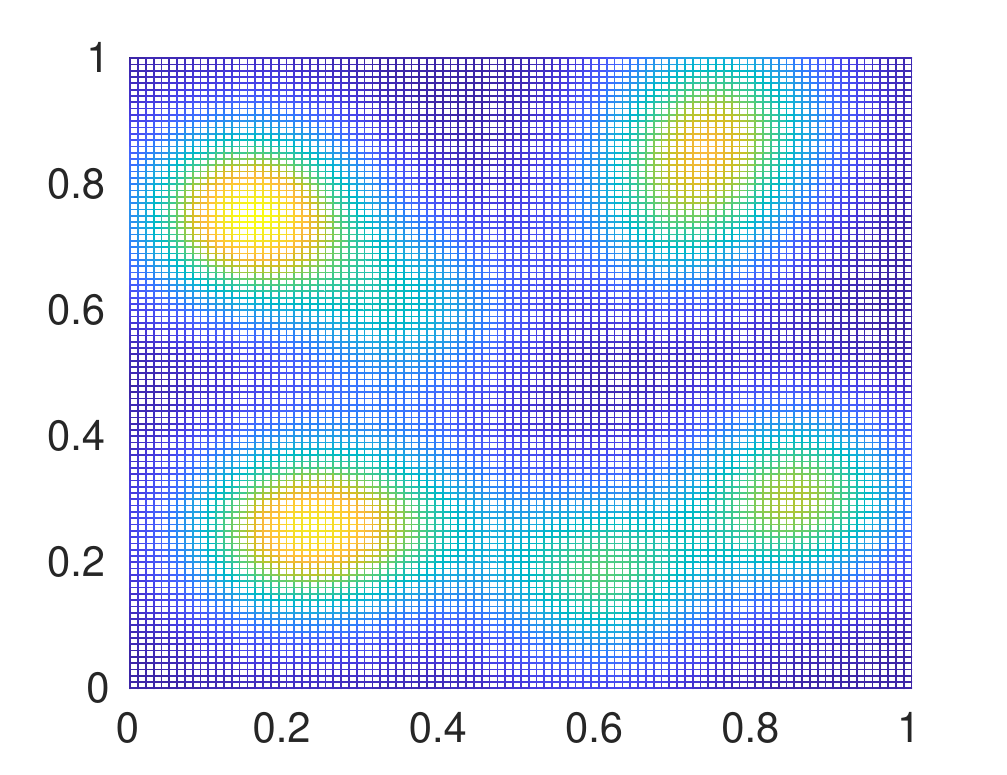}
 \end{subfigure}
 \caption{The scalar field used for verifying the algorithms}
 \label{fig:scalarfield1}
\end{figure}
The initial positions of the mobile sensors were chosen randomly 
and shown in figure \ref{fig:reconstructedfield1}. The partition 
of the region was done by constructing the voronoi cells for 
each mobile sensor. The Voronoi cell of mobile sensor $i$ 
(denoted $\mathcal{Q}_i$) consists of those points which are 
closer to sensor $i$ as compared to all other sensors: 
\begin{equation}
 \mathcal{Q}_i = \{ q \in \mathcal{Q} : \|q-x_i\| \leq \|q-x_j\| 
 , j=1,2,\dots,N; j \neq i \}
 \label{eqn:voronoi}
\end{equation}
For motion control of the sensors, we use a proportional control 
law $u_i = k (x_i - x_{gi})$ where $x_{gi}$ is made to 
switch between all the centres in the region $\mathcal{Q}_i$ 
making sure the condition in lemma \ref{lemma:relaxedpd} is 
satisfied. The control gain $k$ was chosen to be $5$. 
The simulation ran for $16.5$ seconds. The excitation condition 
was achieved in $T = 1.5$ seconds. The reconstructed 
scalar field with algorithm S$1$ is shown in figure 
\ref{fig:reconstructedfield1} on the right and the average 
(across all the mobile sensors) parameter estimation error is 
shown in figure \ref{fig:estimationerror1}. It can be seen that 
the parameters converge exactly to the true values and exact 
reconstruction is achieved.
\begin{figure}
 \begin{subfigure}{0.24\textwidth}
  \centering
  \includegraphics
  [scale=0.4]
  {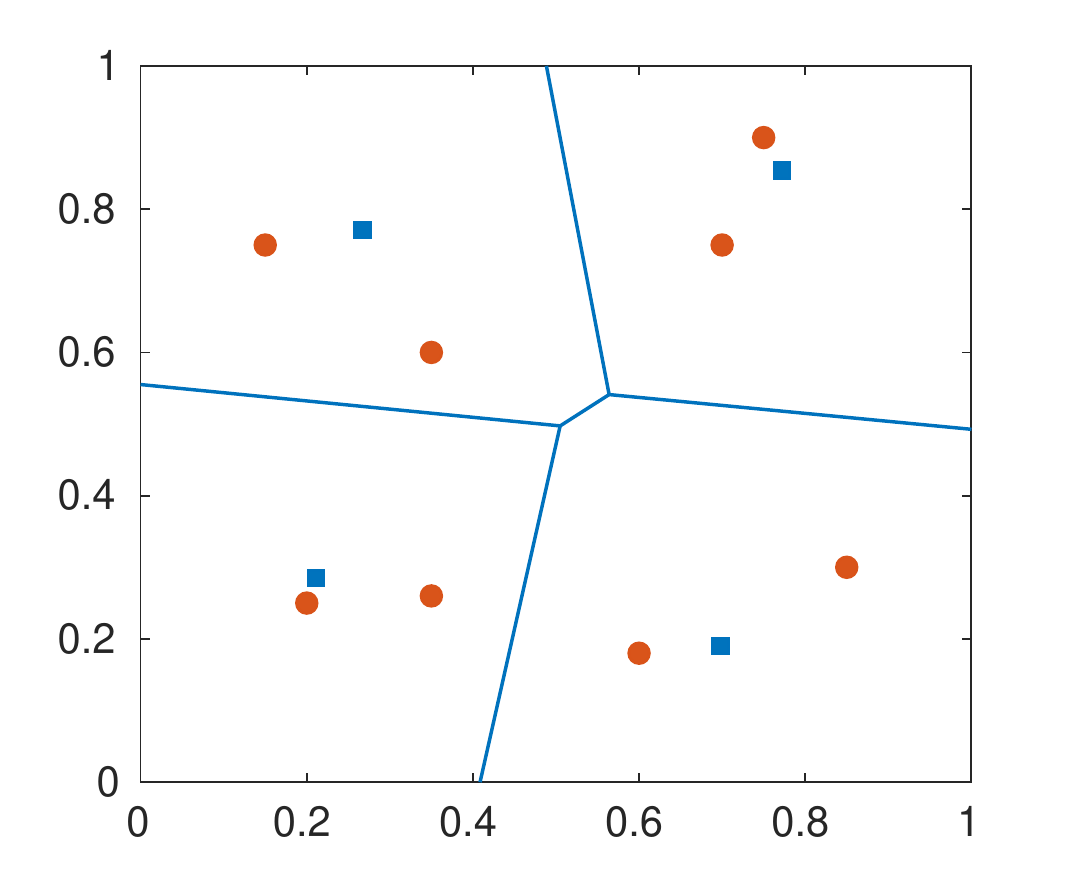}
  \caption{\small Partitions} 
 \end{subfigure}
 \begin{subfigure}{0.24\textwidth}
  \centering
  \includegraphics
  [scale=0.45]
  {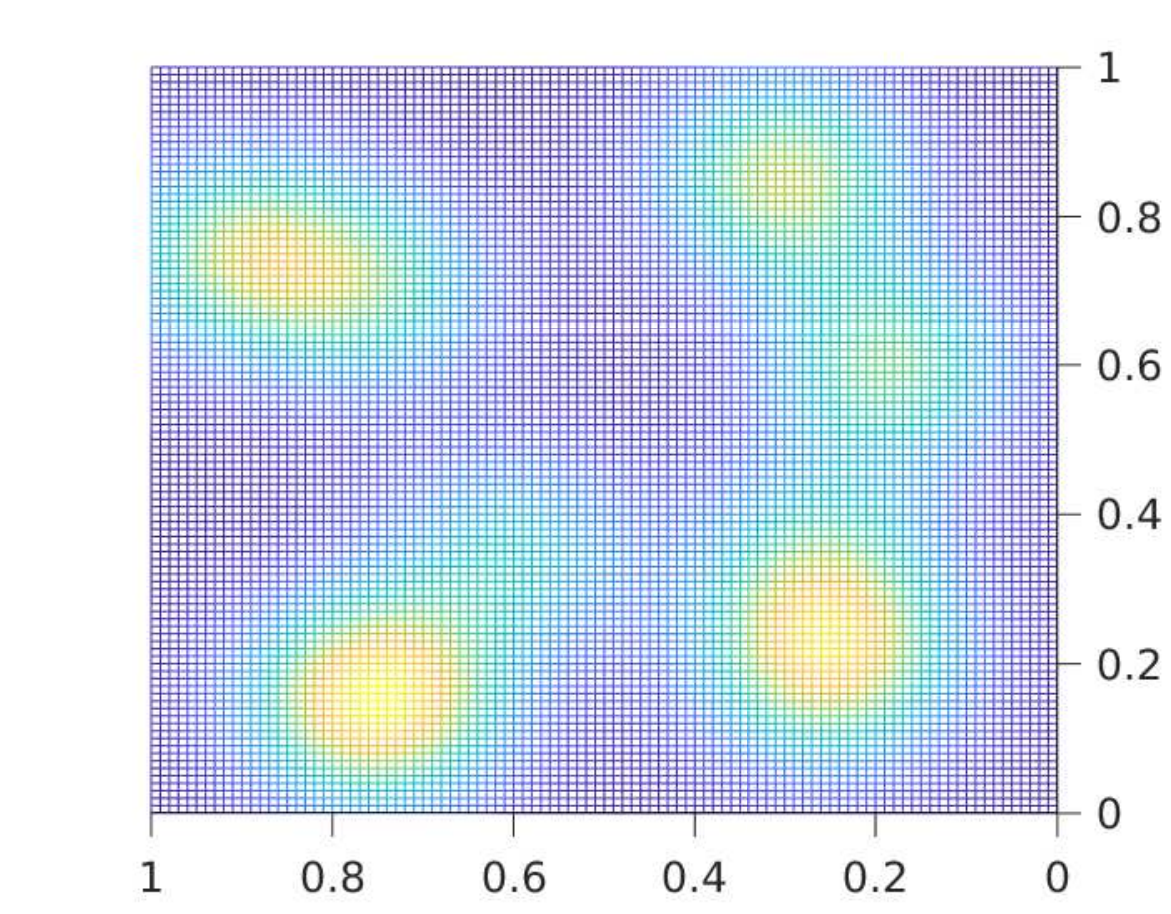}
  \caption{\small Algorithm S$1$.} 
 \end{subfigure}
 \caption{Left: Initial positions (blue squares), corresponding 
 partitions and centres of RBFs (red circles); Right: 
 Reconstructed field using algorithm S$1$.} 
 \label{fig:reconstructedfield1}
\end{figure}
\begin{figure}
 \centering
 \includegraphics
 [width=0.5\textwidth, height=0.2\textheight]
 {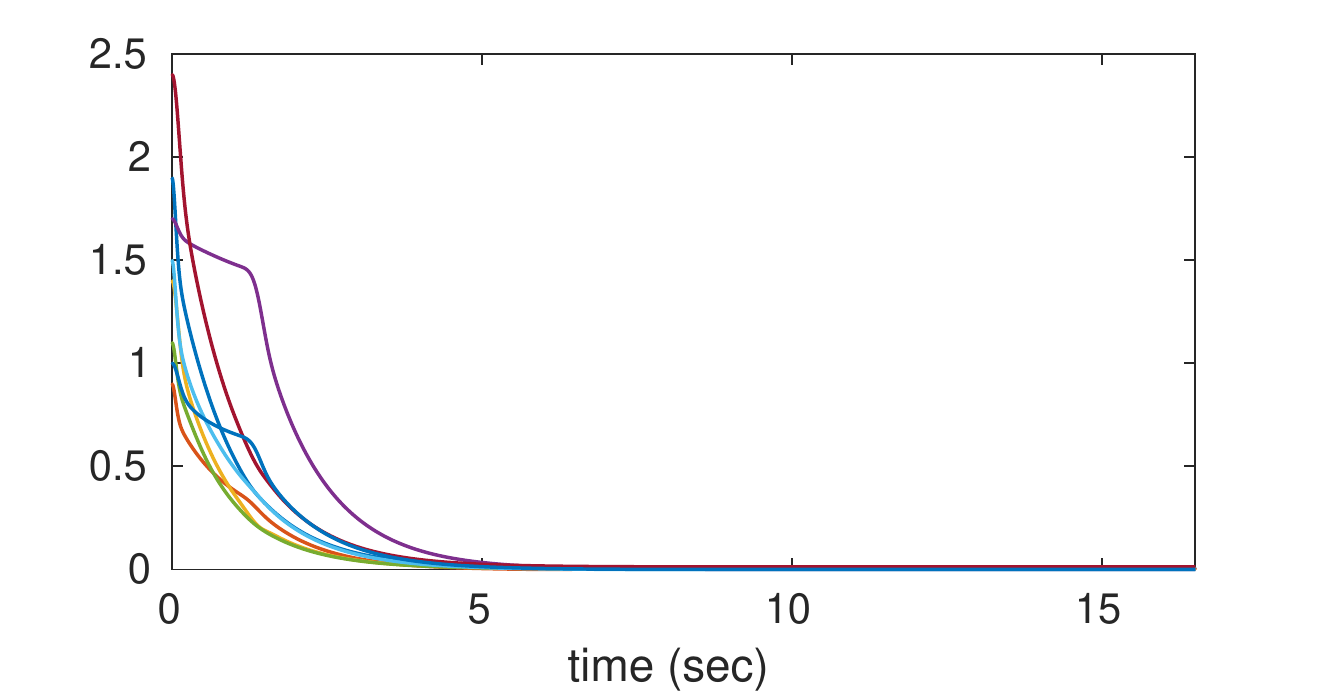}
 \caption{Algorithm S$1$: Average parameter estimation error 
 with time}
 \label{fig:estimationerror1}
\end{figure}
The reconstructed field with algorithm S$2$ and 
algorithm S$3$ are shown in figure \ref{fig:reconstructedfield2}. 
The corresponding estimation errors are shown in figures 
\ref{fig:estimationerror2} and \ref{fig:estimationerror3} 
respectively.
\begin{figure}
 \begin{subfigure}{0.23\textwidth}
  \centering
  \includegraphics
  [scale=0.38]
  {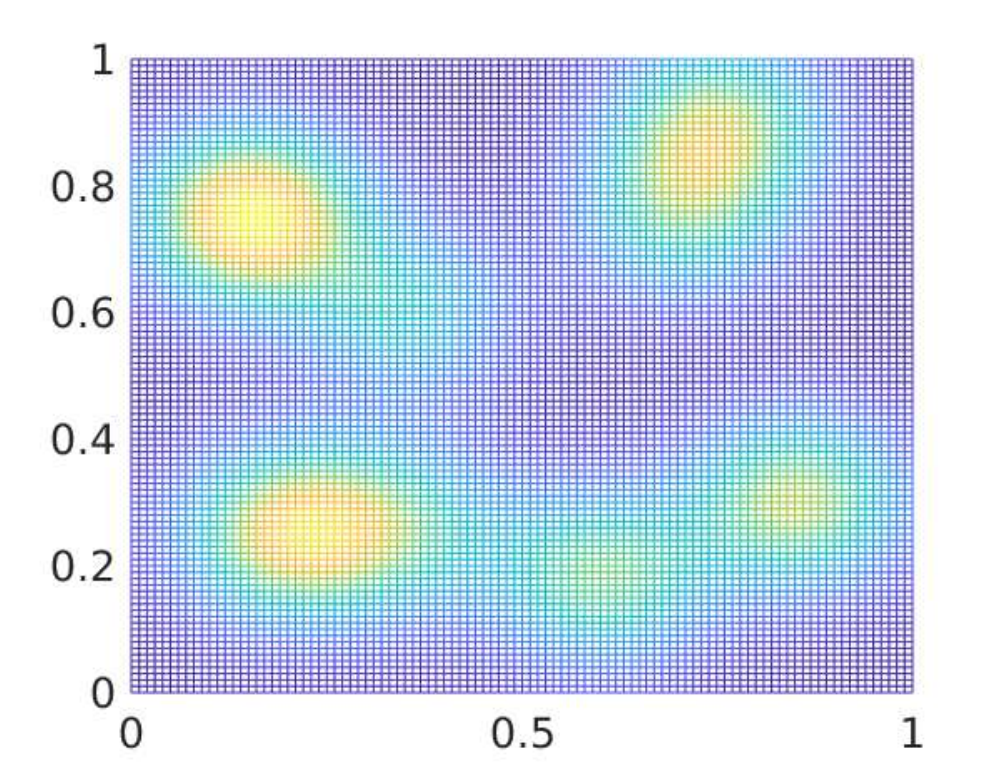}
  \caption{\small Algorithm S$2$.} 
 \end{subfigure}
 \begin{subfigure}{0.24\textwidth}
  \centering
  \includegraphics
  [scale=0.38]
  {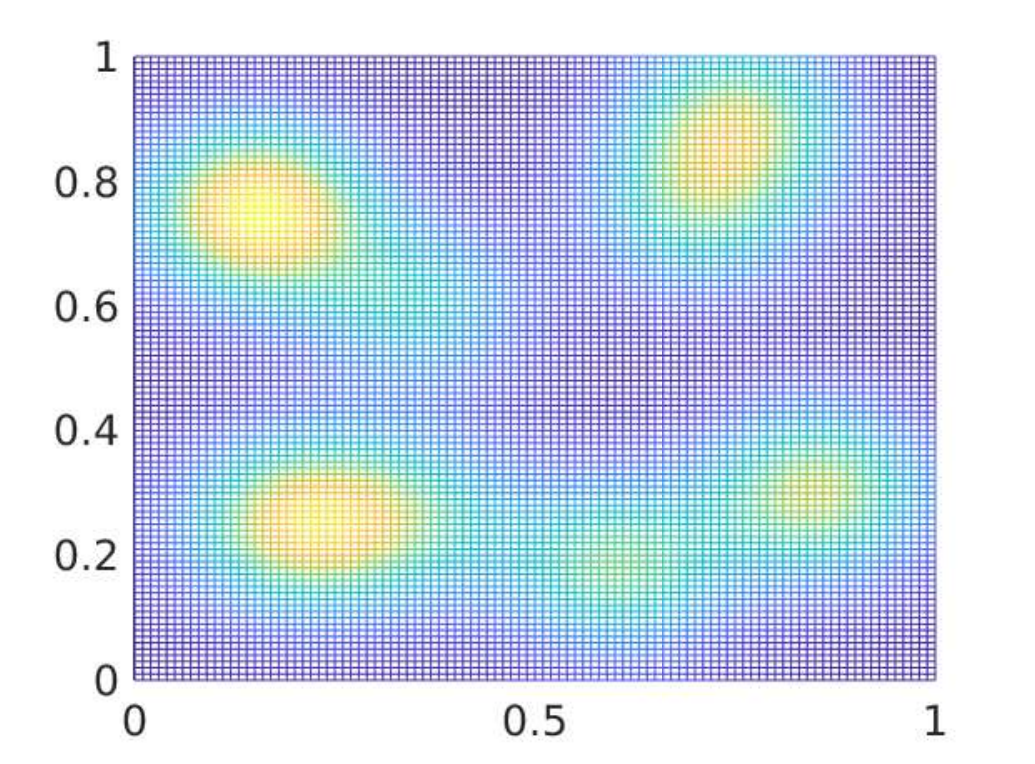}
  \caption{\small Algorithm S$3$.} 
 \end{subfigure}
 \caption{The reconstructed field using algorithm S$2$ and S$3$.}
 \label{fig:reconstructedfield2}
\end{figure}
\begin{figure}
 \centering
 \includegraphics
 [width=0.5\textwidth, height=0.2\textheight]
 {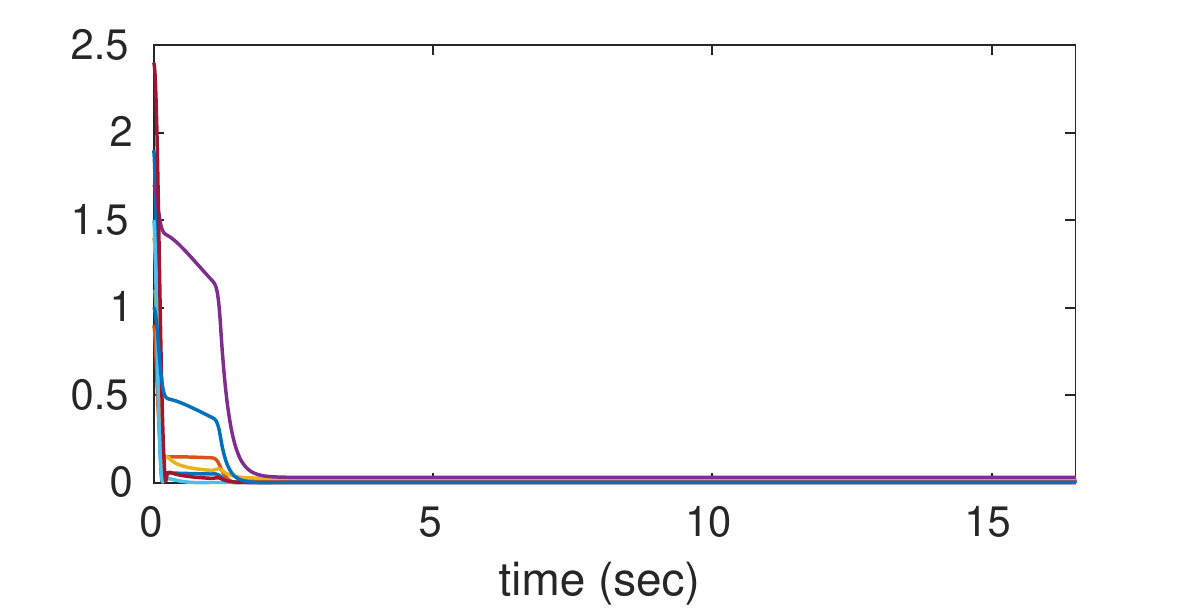}
 \caption{Algorithm S$2$: Average parameter estimation error 
 with time}
 \label{fig:estimationerror2}
\end{figure}
\begin{figure}
 \centering
 \includegraphics
 [width=0.5\textwidth, height=0.2\textheight]
 {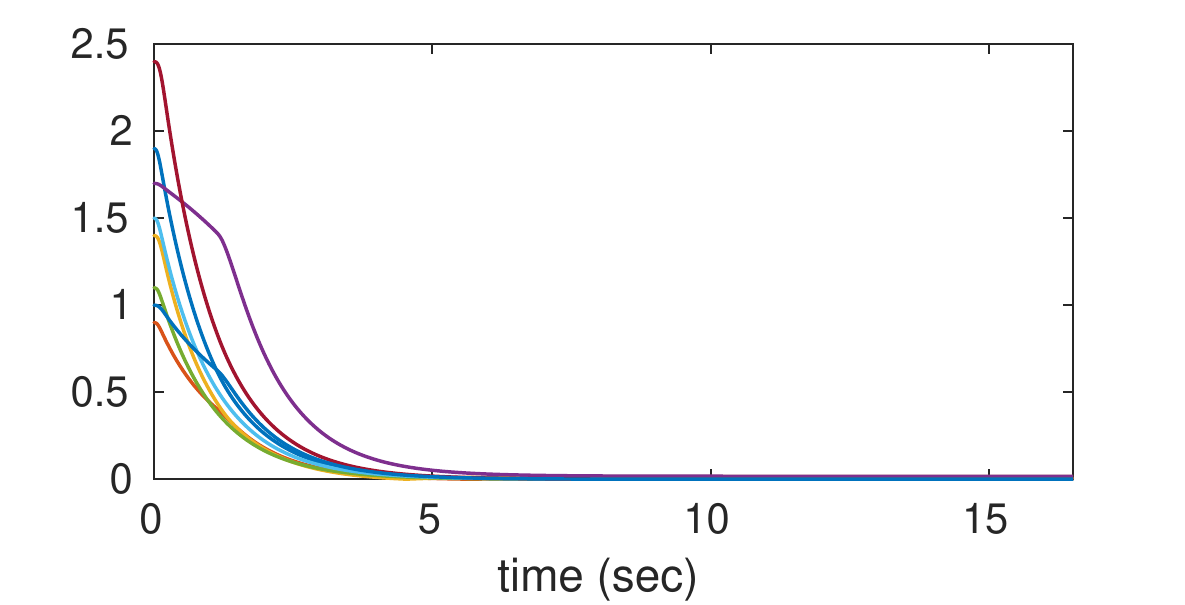}
 \caption{Algorithm S$3$: Average parameter estimation 
 error with time}
 \label{fig:estimationerror3}
\end{figure}
The maximum parameter estimation error using algorithm S$2$ 
was found to be $0.030$ and using the algorithm S$3$ was 
found to be $0.017$. Thus the algorithm S$3$ is seen to 
give better parameter estimates in this case.
\par We also present simulation results where we do not know 
the exact value of the centres of the RBFs (as in section 
\ref{sec:centresnotknown}). We assume we know the centres within 
an accuracy of $\epsilon_c = 0.05$. For this, we add a random 
perturbation (bounded by $\epsilon_c$) to the true centre 
coordinates and use the perturbed centres in the estimation 
algorithm. The reconstructed fields with algorithms S$1$, 
S$2$ and S$3$ are shown in figures 
\ref{fig:reconstructedfield3} and \ref{fig:reconstructedfield4} 
respectively. Table \ref{tab:unknowncentres_maxerror} also 
compares the maximum steady state parameter errors in the 
three cases. As expected, algorithm S$1$ has much lower steady 
state error compared to algorithm S$2$ and 
algorithm S$3$ performs better than algorithm S$2$. It should be 
noted that all the algorithms identify the main features of the 
true field, as seen from the reconstructed field plots.
\begin{figure}
 \begin{subfigure}{0.23\textwidth}
  \centering
  \begin{tabular}{lc}
   \toprule 
   {\small \textit{Algorithm}} & 
   {\small \textit{Max. est. error}} \\
   \midrule 
   \emph{\small S$1$} & $0.16$ \\
   \emph{\small S$2$} & $0.62$ \\
   \emph{\small S$3$} & $0.44$ \\
   \bottomrule
  \end{tabular}
  \caption{\small Max. parameter estimation errors.} 
  \label{tab:unknowncentres_maxerror}
 \end{subfigure}
 \begin{subfigure}{0.23\textwidth}
  \centering
  \includegraphics
  [scale=0.38]
  {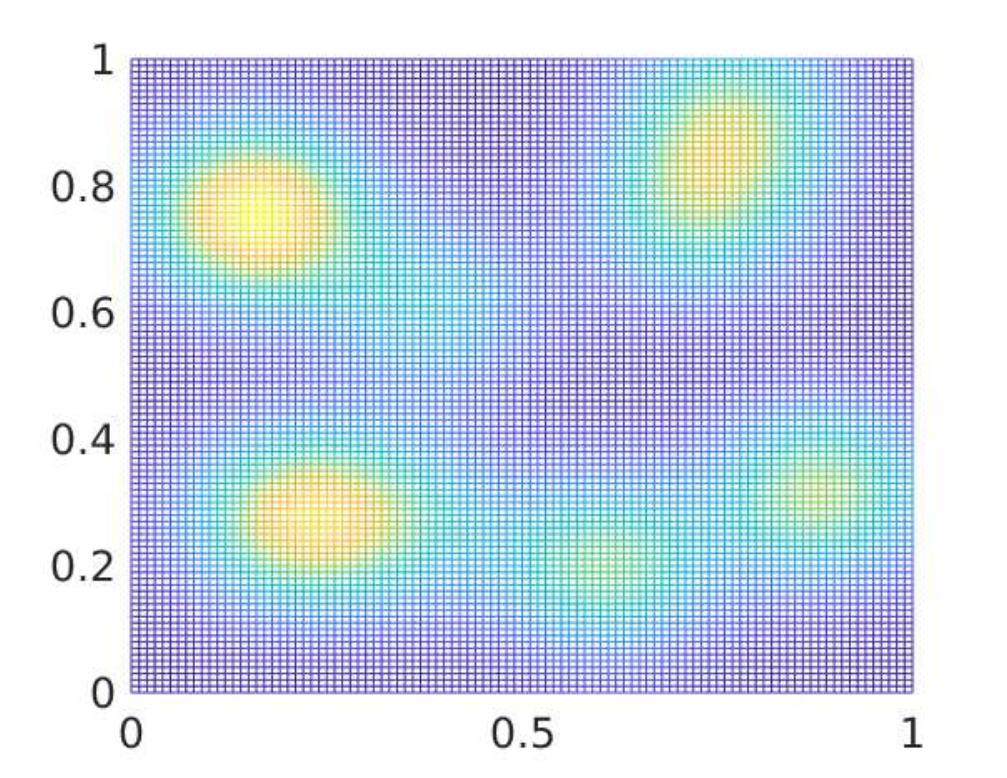}
  \caption{\small Algorithm S$1$.}
 \end{subfigure}
 \caption{Unknown Centres: Max. parameter estimation errors 
 (left) and the reconstructed field using algorithm S$1$ (right).}
 \label{fig:reconstructedfield3}
\end{figure}
\begin{figure}
 \begin{subfigure}{0.23\textwidth}
  \centering
  \includegraphics
  [scale=0.38]
  {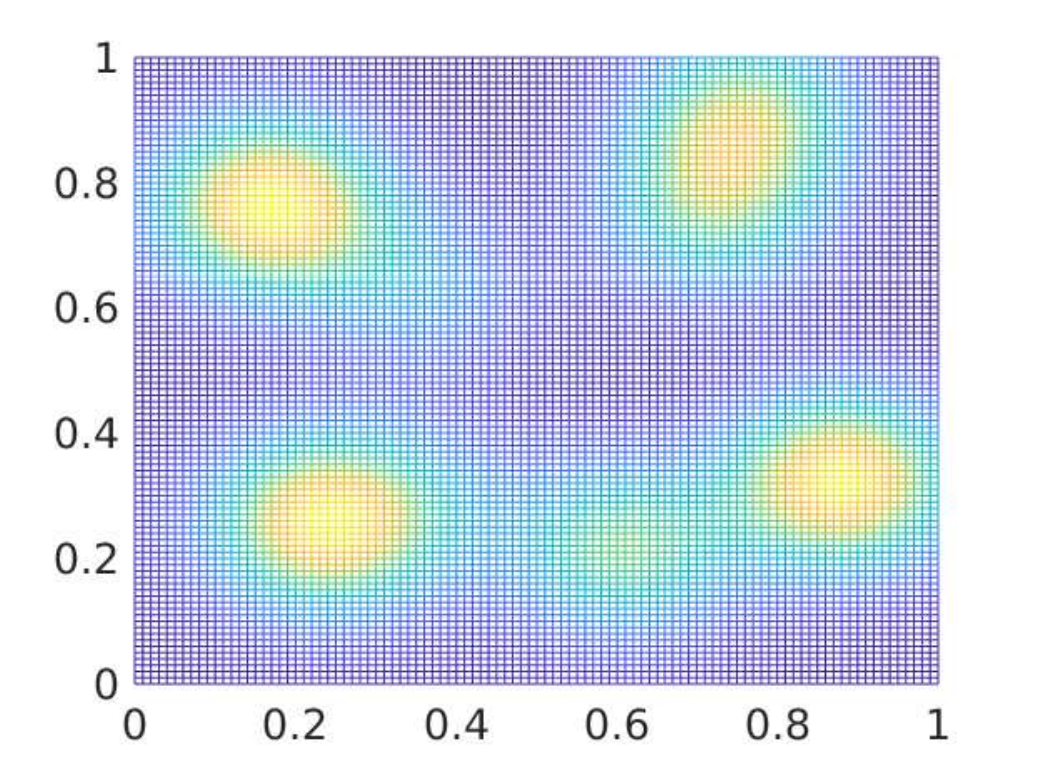}
  \caption{\small Algorithm S$2$.}
 \end{subfigure}
 \begin{subfigure}{0.24\textwidth}
  \centering
  \includegraphics
  [scale=0.38]
  {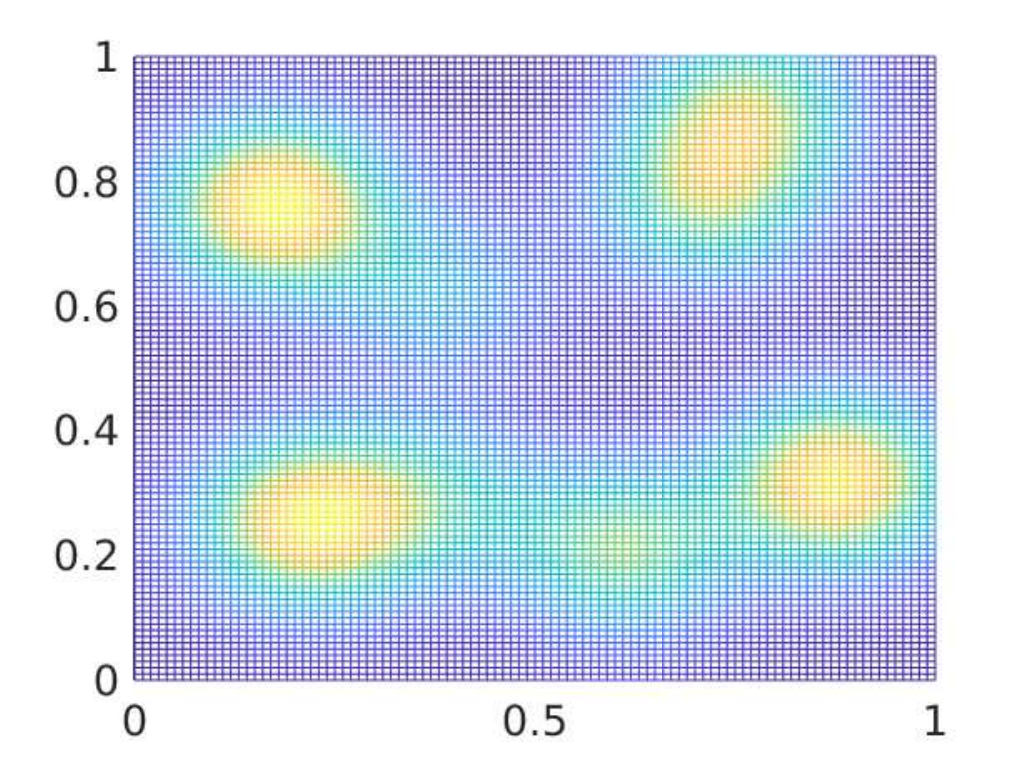}
  \caption{\small Algorithm S$3$.}
 \end{subfigure}
 \caption{Unknown Centres: Reconstructed field.}
 \label{fig:reconstructedfield4}
\end{figure}

\subsection{Fully unknown scalar field}
\par Now we test the estimation algorithms on a more general 
scalar field which is not a linear combination of RBFs. 
For this we consider the continuous scalar field given by 
\begin{align*}
 \phi(x,y) &= 3x^2 e^{\frac{-(x-0.7)^2 - (y-0.7)^2}{0.05}} + 
 e^{\frac{-(x-0.4)^2-(y-0.4)^2}{0.06}} \\ 
 & \qquad + \frac{1}{3} e^{\frac{-(x-0.2)^2-(y-0.2)^2}{0.08}}.
\end{align*}
over the unit square region $\mathcal{Q}$.
A plot of $\phi(\cdot)$ is shown in figure \ref{fig:scalarfield2}.
We use $N=5$ mobile sensors with the partitions $\mathcal{Q}_i$ 
determined as follows: We first run a uniform coverage algorithm 
(coverage algorithm presented in \cite{cortes2004coverage} with 
a uniform density function $\phi(q) \equiv 1$). This makes the 
mobile sensors uniformly spread out in the region $\mathcal{Q}$. 
We then compute the voronoi partition \eqref{eqn:voronoi} of 
the sensors and use it as the required partition $\mathcal{Q}_i$.

\begin{figure}
 \begin{subfigure}{0.24\textwidth}
  \centering
  \includegraphics
  [scale=0.4]
  {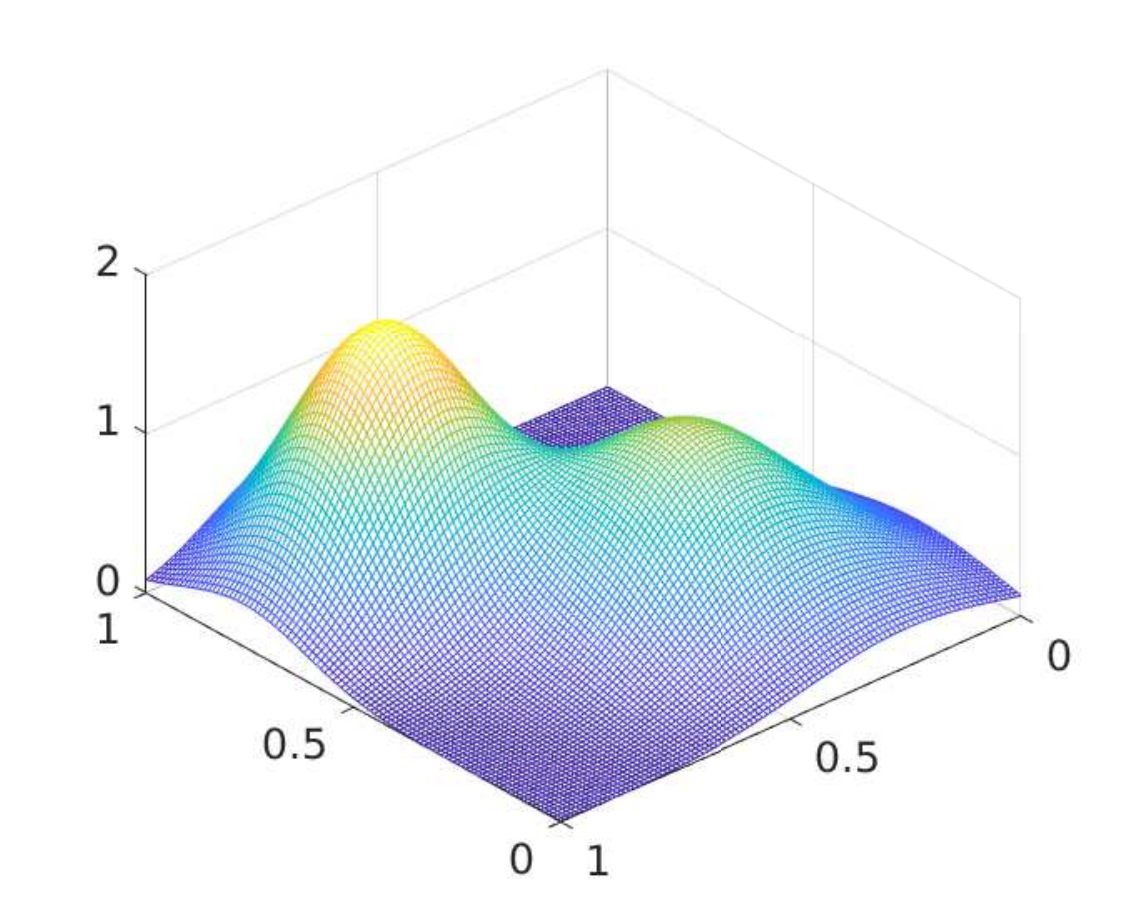}
 \end{subfigure}
 \begin{subfigure}{0.24\textwidth}
  \centering
  \includegraphics
  [scale=0.4]
  {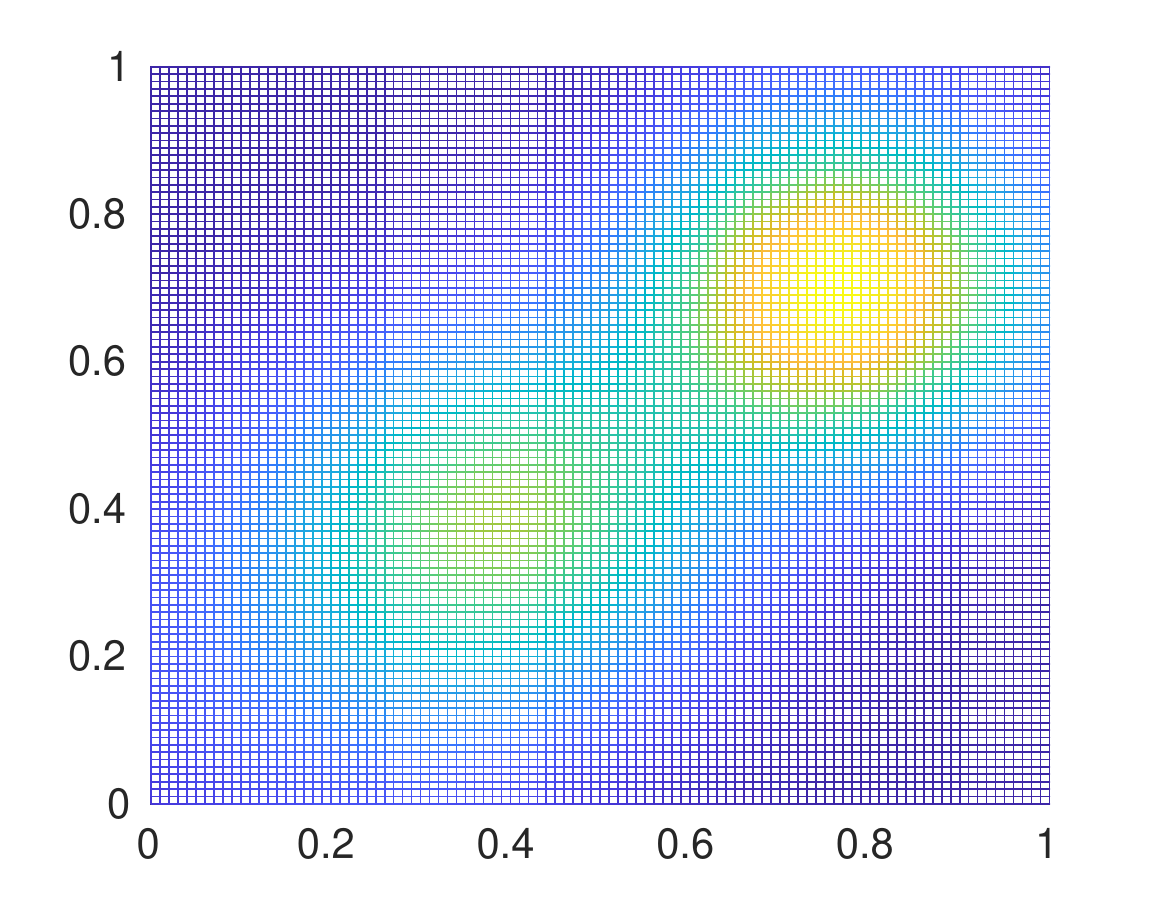}
 \end{subfigure}
 \caption{The scalar field $\phi(x,y)$ used in the simulation}
 \label{fig:scalarfield2}
\end{figure}
We first show the results for approximating the field 
$\phi(\cdot)$ with $p=100$ Gaussian RBFs. The centres of 
the Gaussian are arranged in a uniform grid over the region 
$\mathcal{Q}$. The reconstructed field plots for two values of 
$\sigma_i$ (standard deviation of the Gaussian RBFs) are shown 
in figures \ref{fig:reconstructedfield5}, 
\ref{fig:reconstructedfield6} and \ref{fig:reconstructedfield7} 
with the three algorithms.
\begin{figure}
 \begin{subfigure}{0.24\textwidth}
  \centering
  \includegraphics
  [scale=0.4]
  {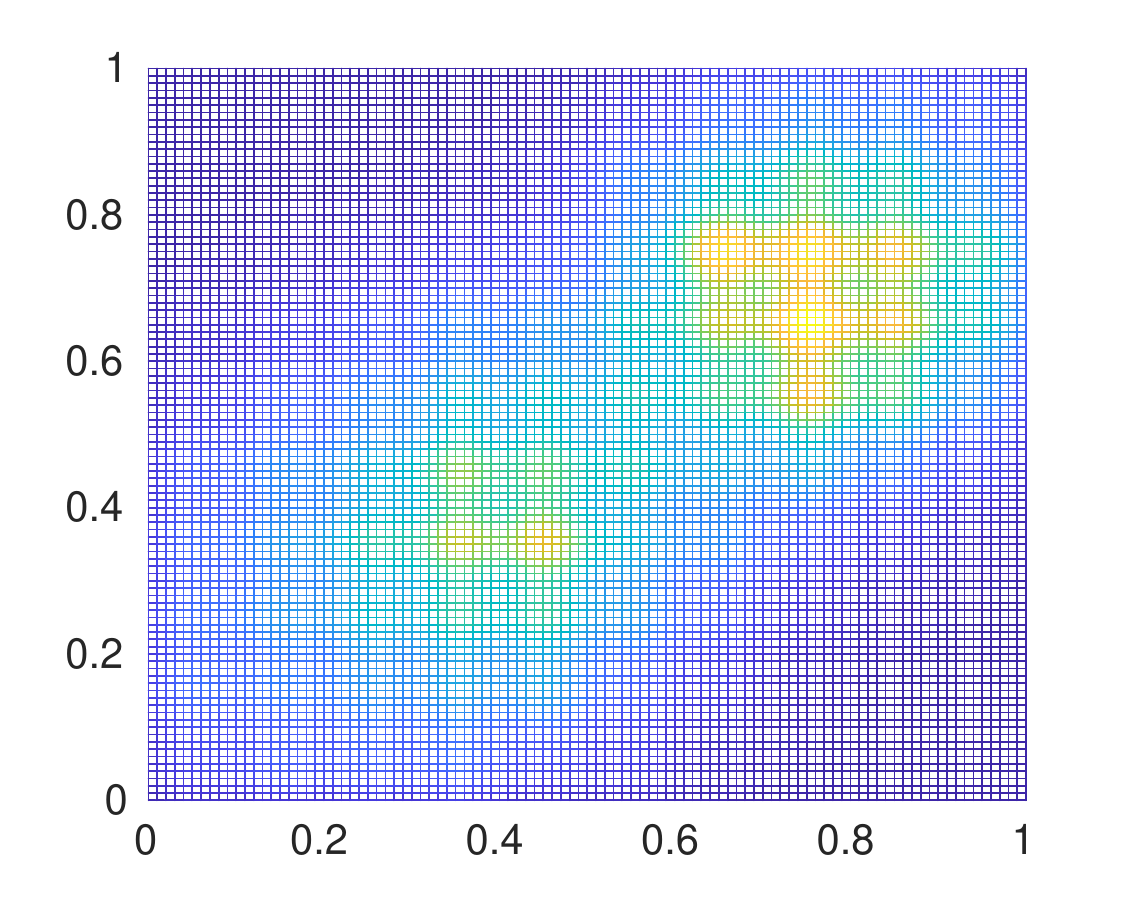}
  \caption{\small $\sigma_i = 0.04$.}
 \end{subfigure}
 \begin{subfigure}{0.24\textwidth}
  \centering
  \includegraphics
  [scale=0.4]
  {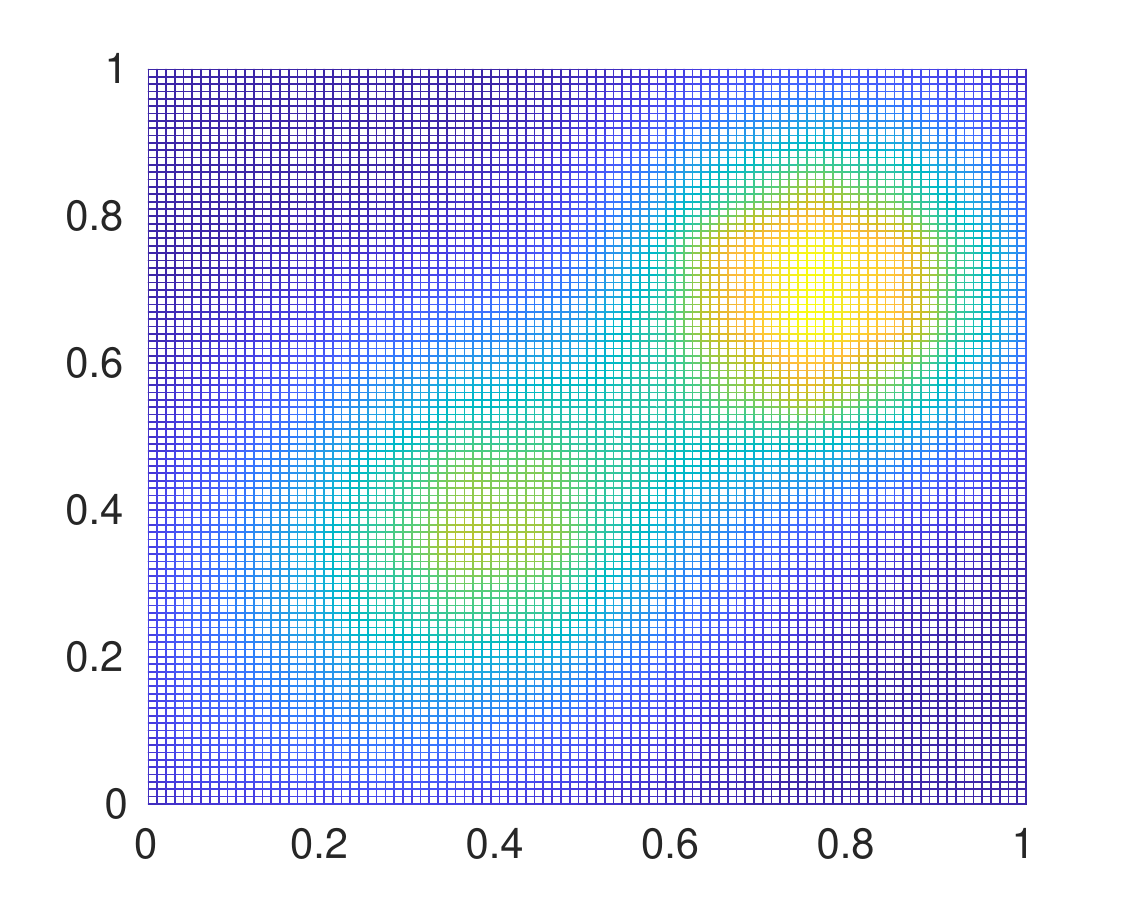}
  \caption{\small $\sigma_i = 0.05$.}
 \end{subfigure}
 \caption{Reconstructed field ($p=100$) with algorithm S$1$.}
 \label{fig:reconstructedfield5}
\end{figure}
\begin{figure}
 \begin{subfigure}{0.24\textwidth}
  \centering
  \includegraphics
  [scale=0.4]
  {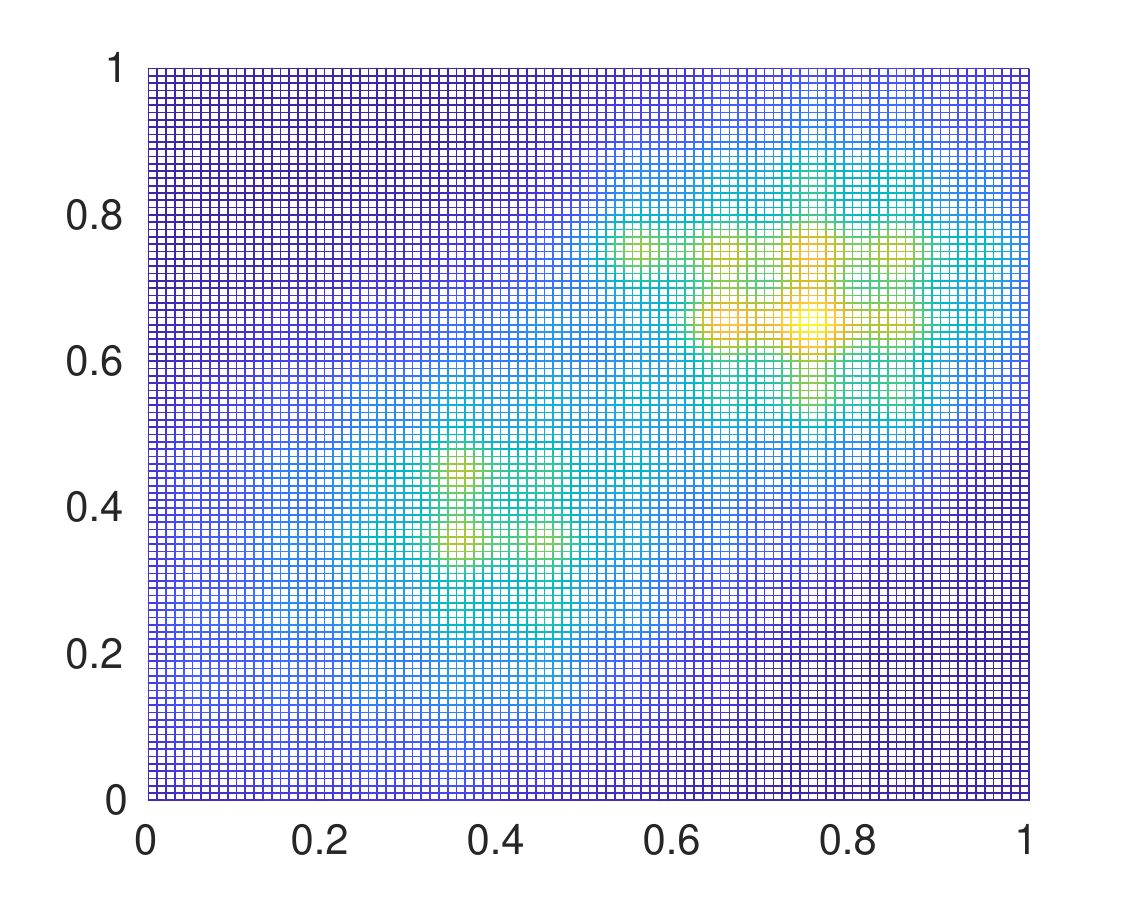}
  \caption{\small $\sigma_i = 0.04$.}
 \end{subfigure}
 \begin{subfigure}{0.24\textwidth}
  \centering
  \includegraphics
  [scale=0.4]
  {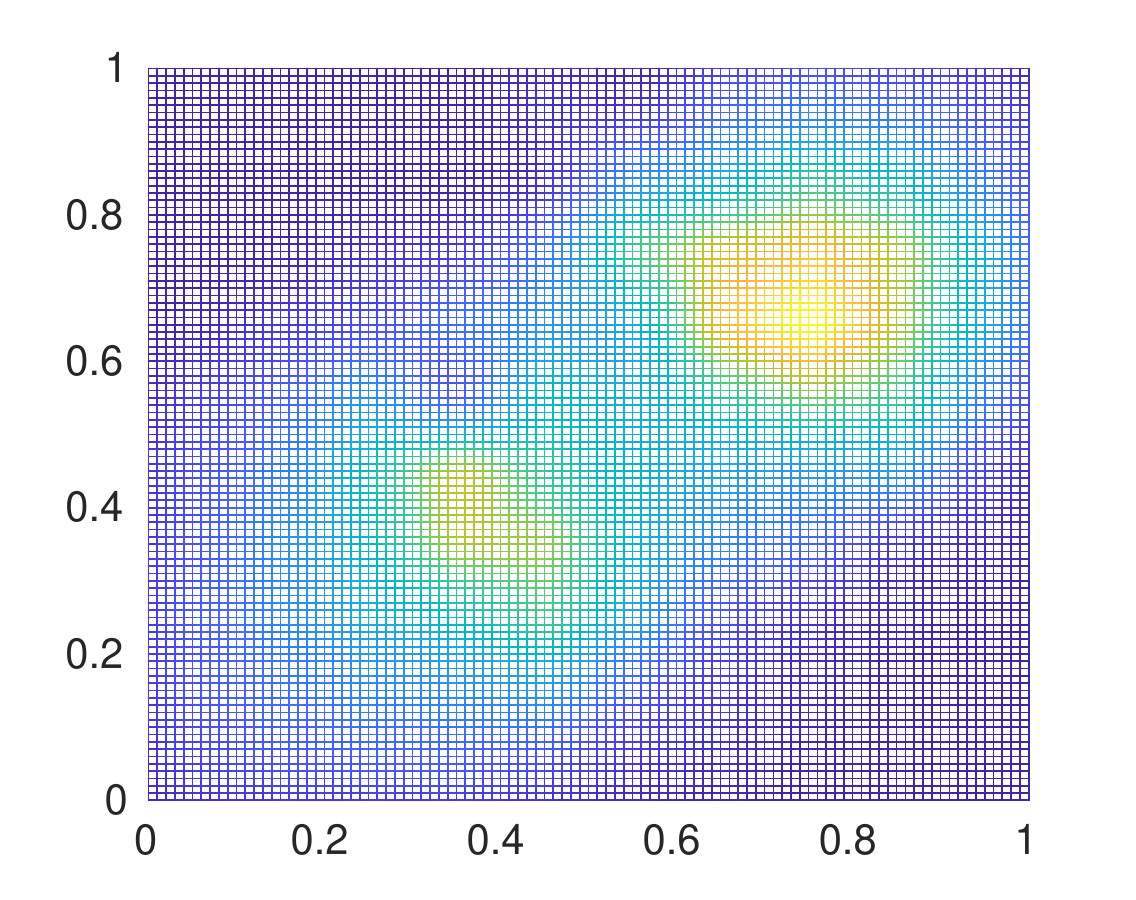}
  \caption{\small $\sigma_i = 0.05$.}
 \end{subfigure}
 \caption{Reconstructed field ($p=100$) with algorithm S$2$.}
 \label{fig:reconstructedfield6}
\end{figure}
\begin{figure}
 \begin{subfigure}{0.24\textwidth}
  \centering
  \includegraphics
  [scale=0.4]
  {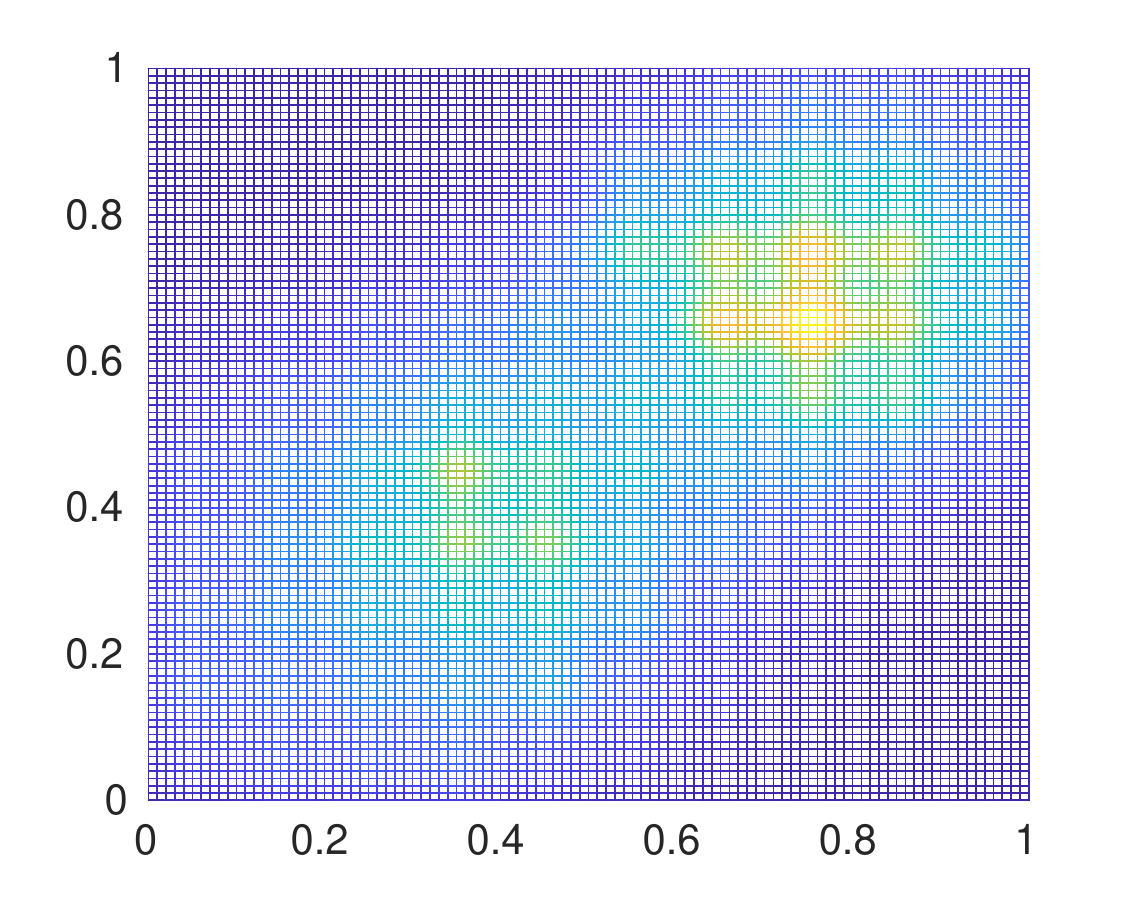}
  \caption{\small $\sigma_i = 0.04$.}
 \end{subfigure}
 \begin{subfigure}{0.24\textwidth}
  \centering
  \includegraphics
  [scale=0.4]
  {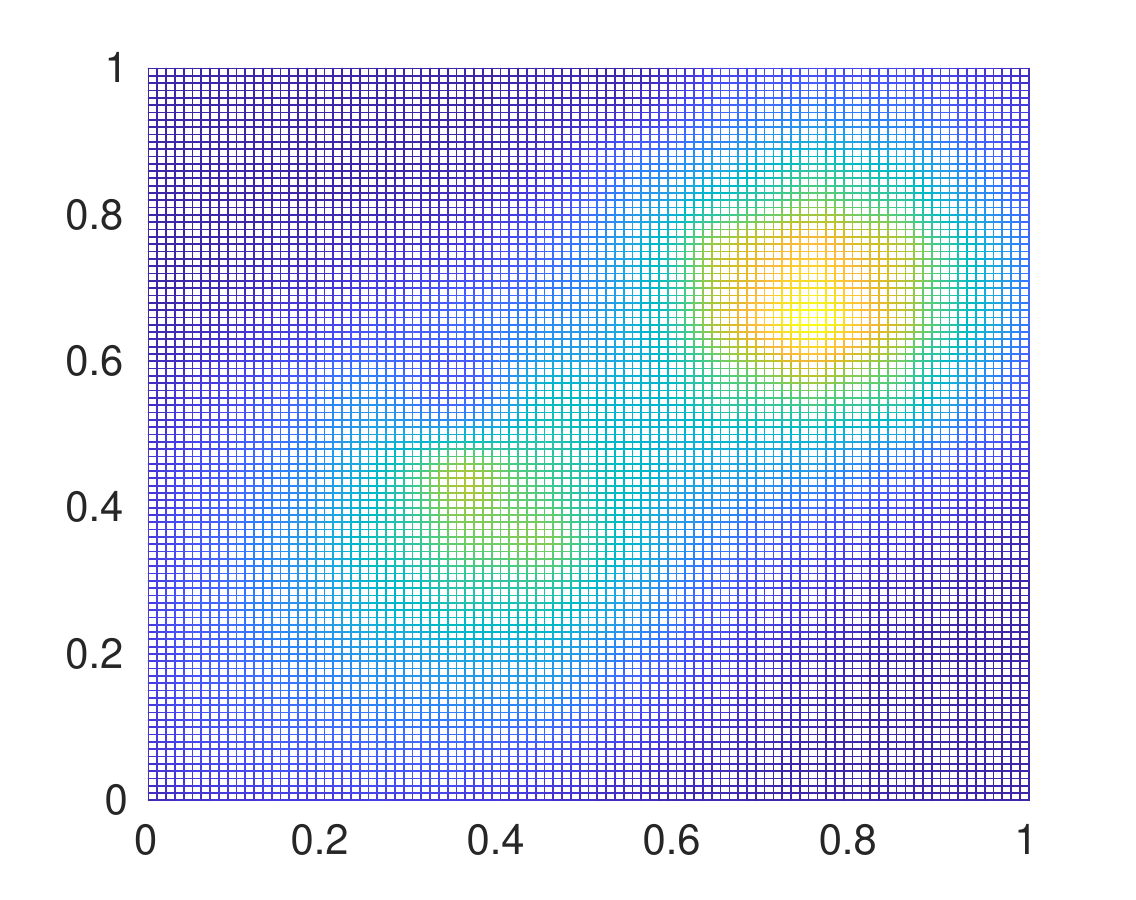}
  \caption{\small $\sigma_i = 0.05$.}
 \end{subfigure}
 \caption{Reconstructed field ($p=100$) with 
 algorithm S$3$.}
 \label{fig:reconstructedfield7}
\end{figure}
To compare the various algorithms, we use the integral error 
(see theorem \ref{thm:approx1}) 
\[
 \|e\|_2 = 
 \int_{\mathcal{Q}} |\phi(q)-\mathcal{K}(q)^{\trp}\hat{a}| dq
\]
where $\hat{a}$ is the final parameter estimate obtained from 
the given algorithm. The integral error for approximation of 
$\phi(\cdot)$ using $p=100$ parameters is shown in table 
\ref{tab:compare_np100}. The table also shows the time $T$ in 
seconds at which the excitation (positive definiteness) 
condition is achieved. The total runtime of the estimation 
algorithms was $T+20$ seconds.
\begin{table}
 \centering
 \begin{tabular}{lcclcc}
 \toprule
  ${\sigma_i=0.04}$ & $T$ (sec) & $\|e\|_2$ & ${\sigma_i=0.05}$ & 
  $T$ (sec) & $\|e\|_2$ \\
  \cmidrule(r){1-3} \cmidrule(l){4-6}
  \emph{Algorithm S$1$} & $3.1$ & $0.045$ & \emph{Algorithm S$1$} & 
  $3.9$ &$0.012$ \\
  \emph{Algorithm S$2$} & $3.1$ & $0.054$ & \emph{Algorithm S$2$} & 
  $3.7$ &$0.053$ \\
  \emph{Algorithm S$3$} & $3.1$ & $0.048$ & \emph{Algorithm S$3$} 
  & $3.7$ & $0.028$ \\
  \bottomrule
 \end{tabular}
 \caption{Comparison of algorithms for $p=100$ parameters.}
 \label{tab:compare_np100}
\end{table}

\par The reconstructed field plots for $p=196$ parameters is 
shown in figures \ref{fig:reconstructedfield8}, 
\ref{fig:reconstructedfield9} and \ref{fig:reconstructedfield10} 
with the three algorithms. The comparison of various algorithms 
is given in table \ref{tab:compare_np196}.
\begin{figure}
 \begin{subfigure}{0.24\textwidth}
  \centering
  \includegraphics
  [scale=0.4]
  {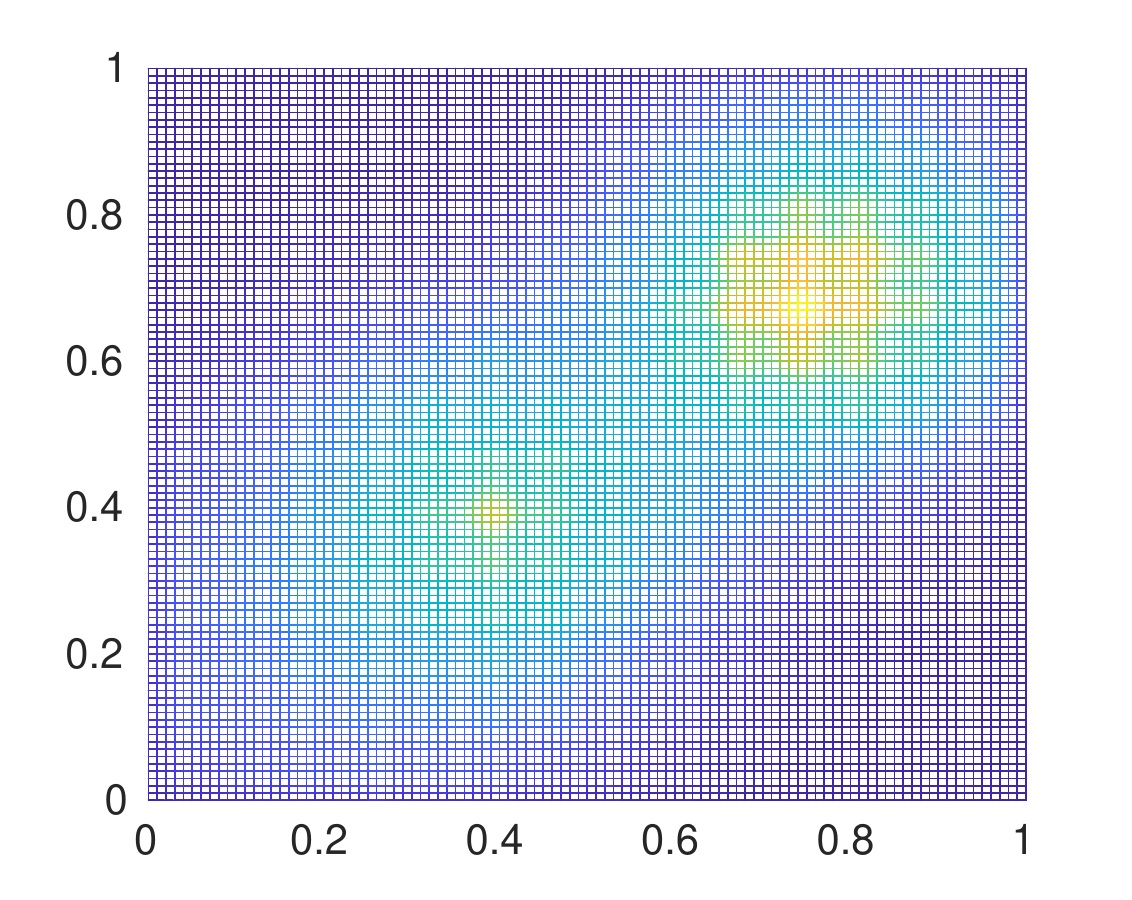}
  \caption{\small $\sigma_i = 0.03$.}
 \end{subfigure}
 \begin{subfigure}{0.24\textwidth}
  \centering
  \includegraphics
  [scale=0.4]
  {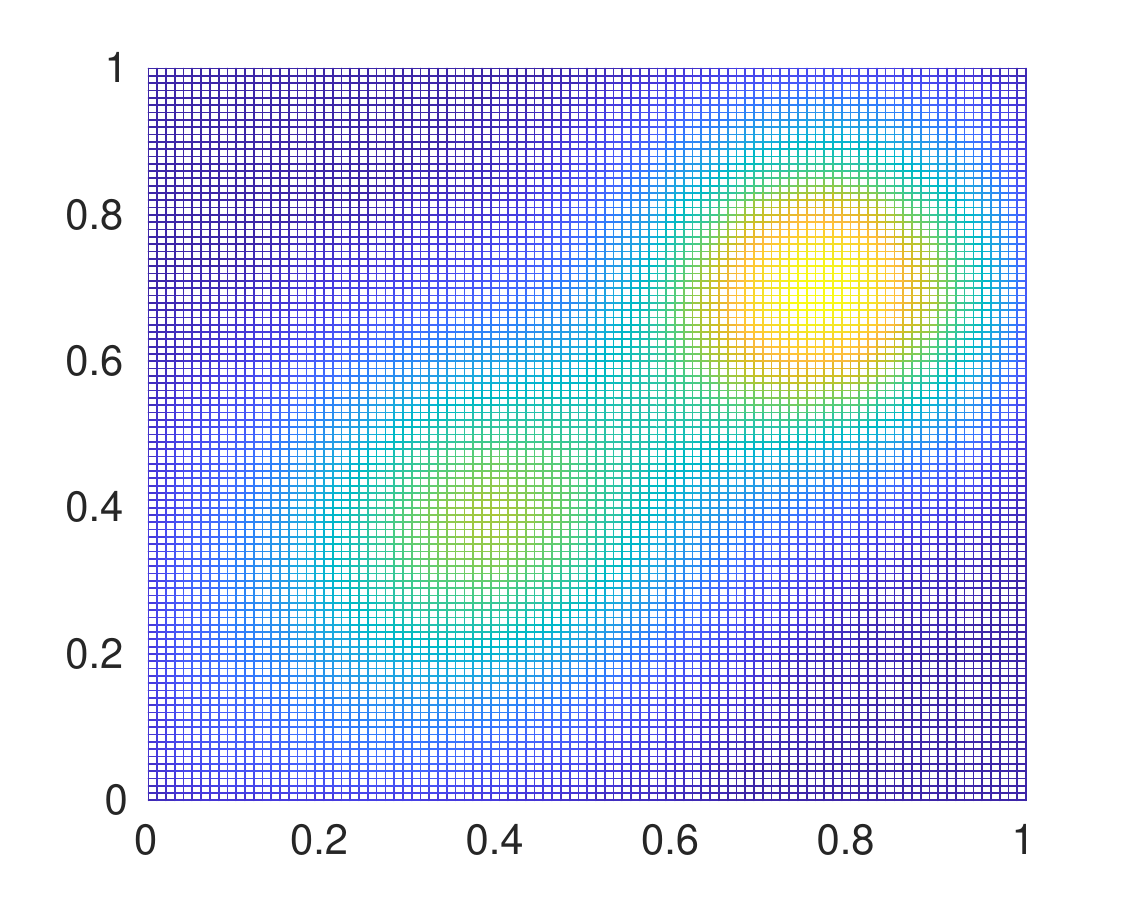}
  \caption{\small $\sigma_i = 0.04$.}
 \end{subfigure}
 \caption{Reconstructed field ($p=196$) with algorithm S$1$.}
 \label{fig:reconstructedfield8}
\end{figure}
\begin{figure}
 \begin{subfigure}{0.24\textwidth}
  \centering
  \includegraphics
  [scale=0.4]
  {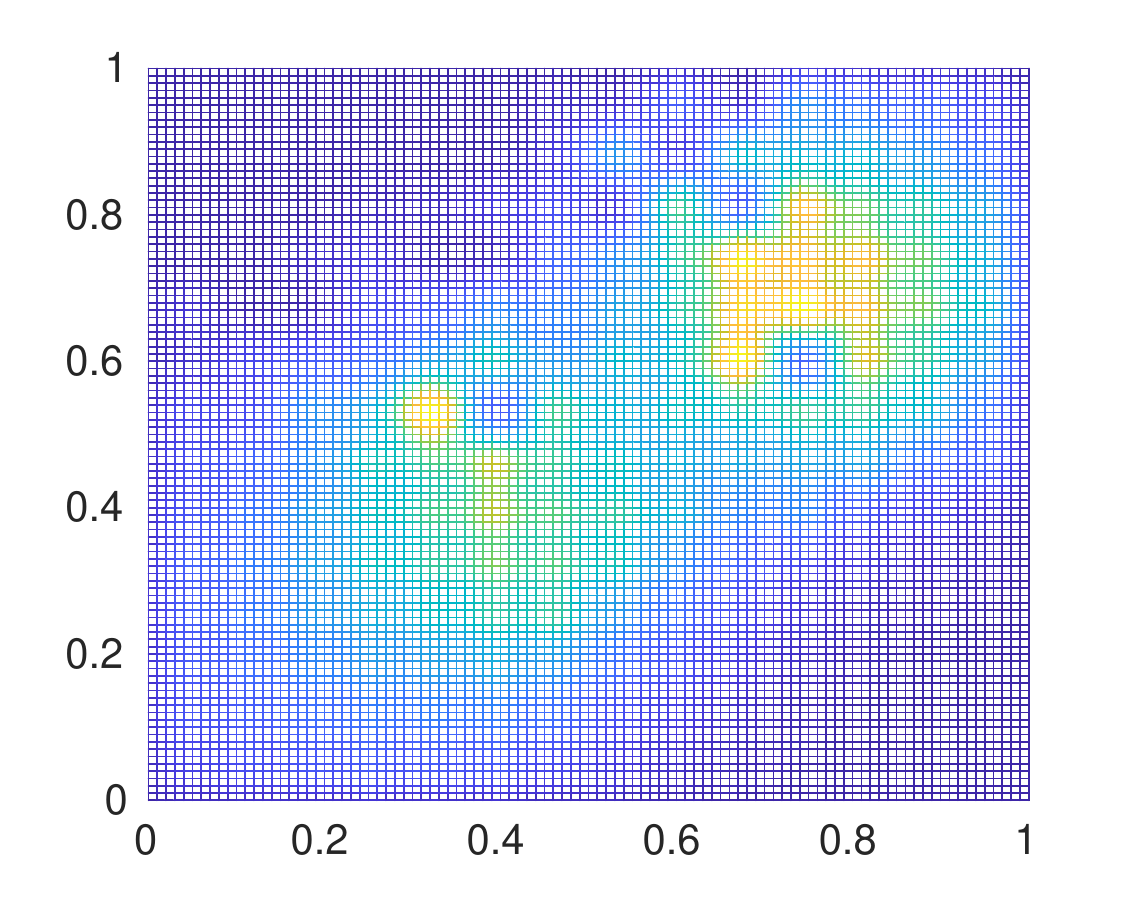}
  \caption{\small $\sigma_i = 0.03$.}
 \end{subfigure}
 \begin{subfigure}{0.24\textwidth}
  \centering
  \includegraphics
  [scale=0.4]
  {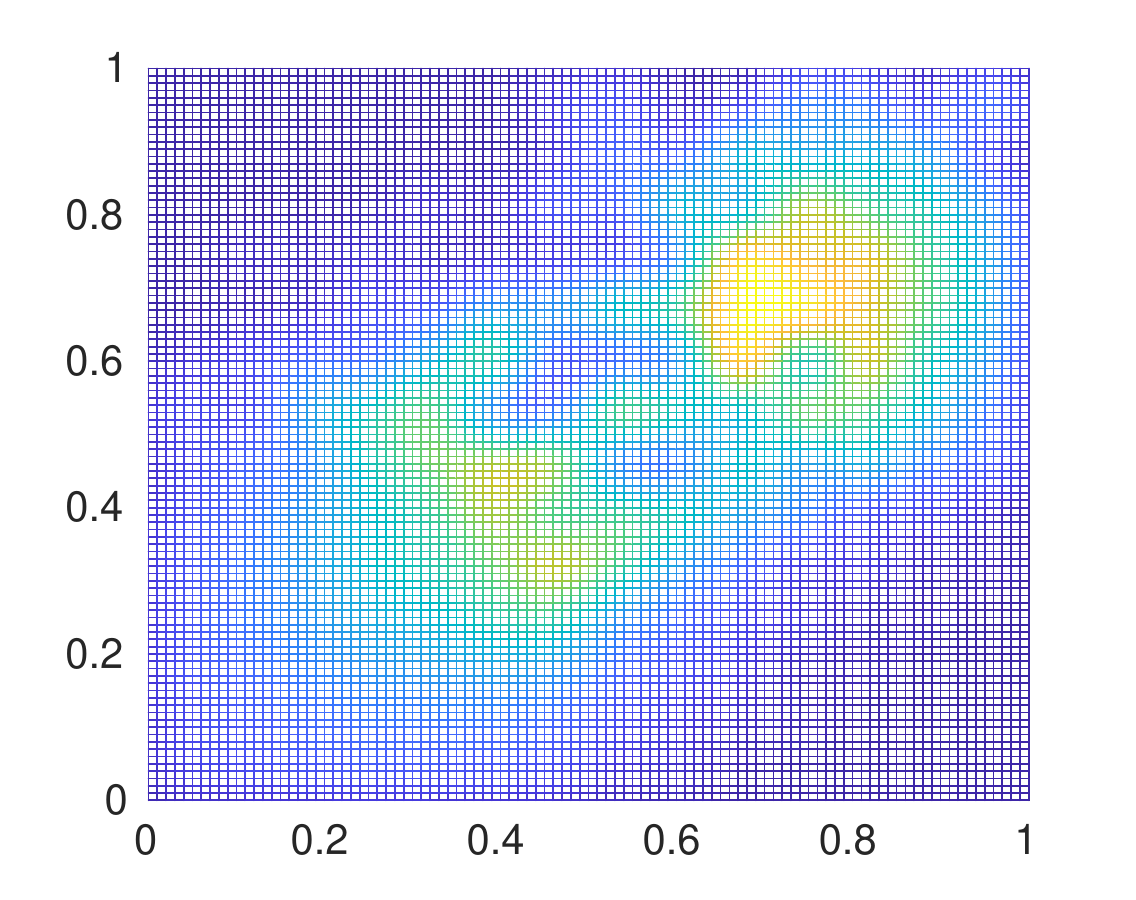}
  \caption{\small $\sigma_i = 0.04$.}
 \end{subfigure}
 \caption{Reconstructed field ($p=196$) with algorithm S$2$.}
 \label{fig:reconstructedfield9}
\end{figure}
\begin{figure}
 \begin{subfigure}{0.24\textwidth}
  \centering
  \includegraphics
  [scale=0.4]
  {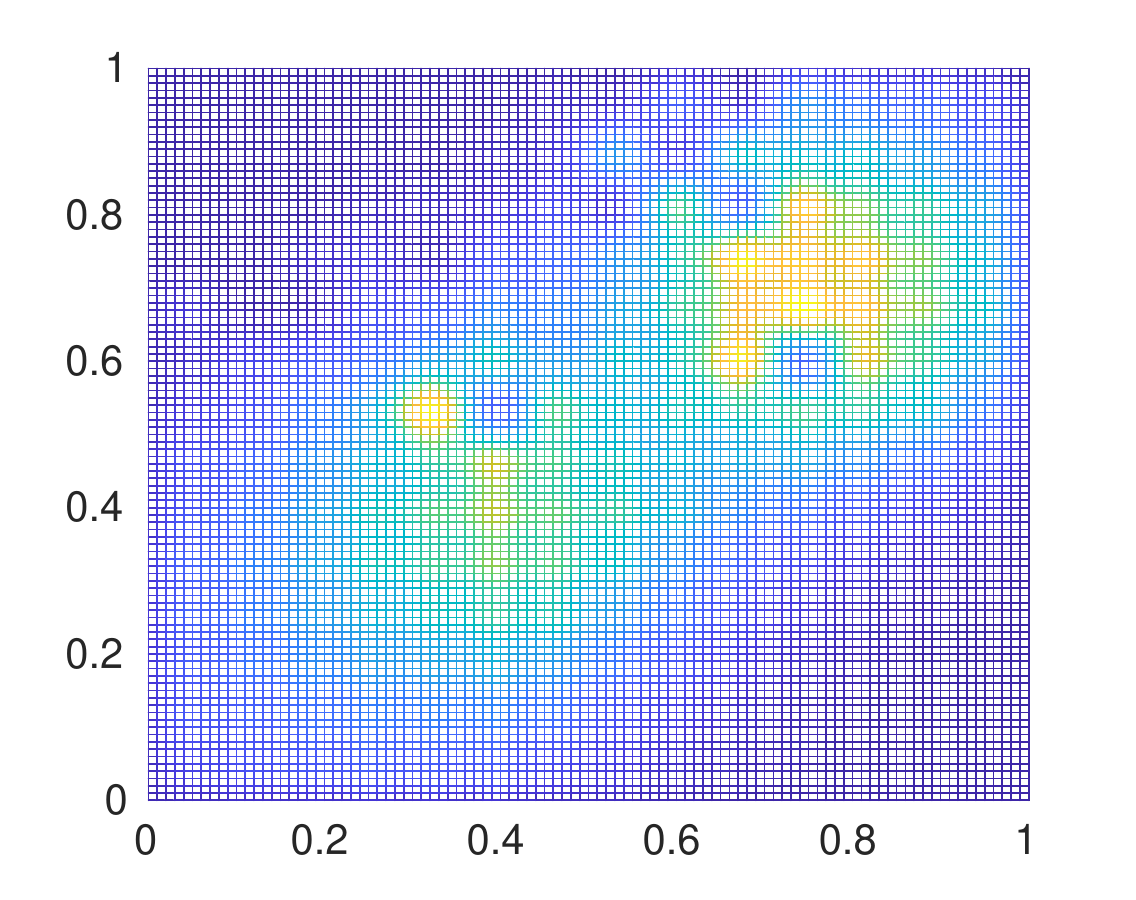}
  \caption{\small $\sigma_i = 0.03$.}
 \end{subfigure}
 \begin{subfigure}{0.24\textwidth}
  \centering
  \includegraphics
  [scale=0.4]
  {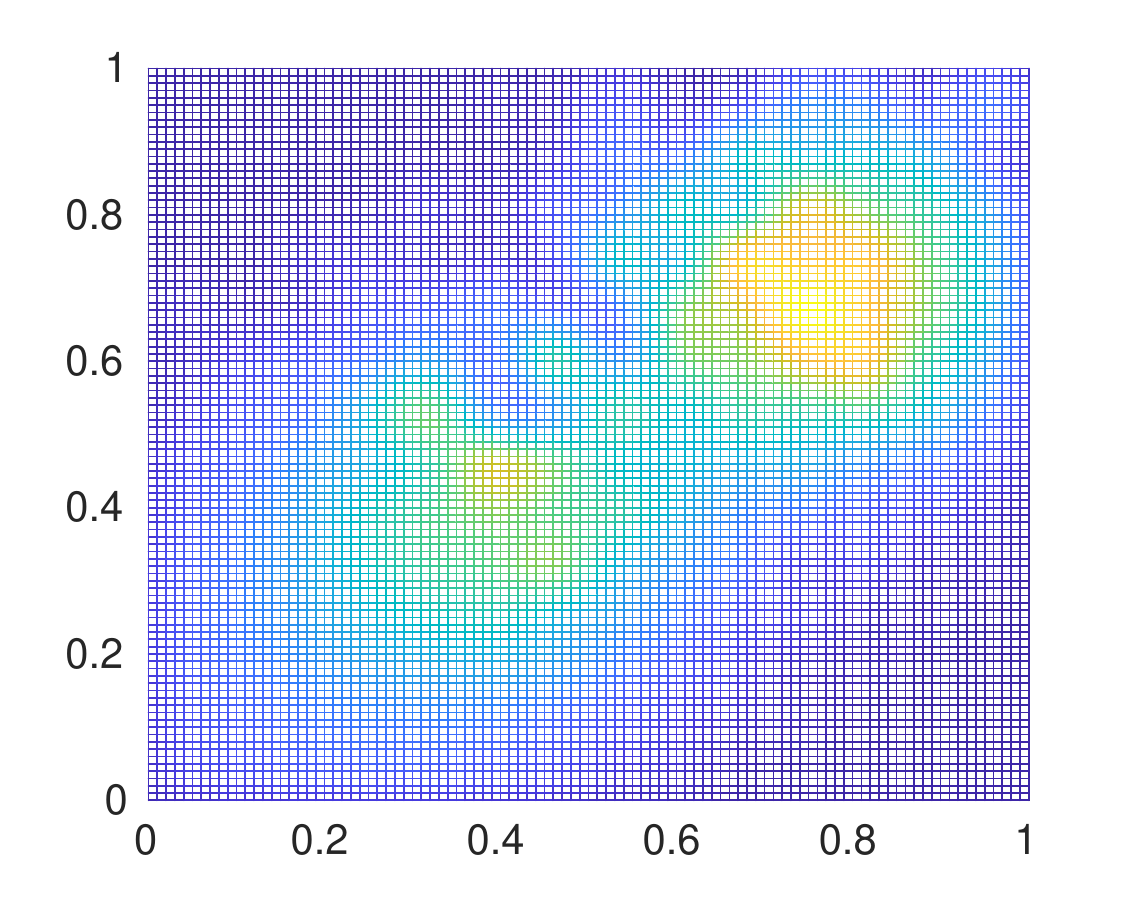}
  \caption{\small $\sigma_i = 0.04$.}
 \end{subfigure}
 \caption{Reconstructed field ($p=196$) with 
 algorithm S$3$.}
 \label{fig:reconstructedfield10}
\end{figure}
\begin{table}
 \centering
 \begin{tabular}{lcclcc}
  \toprule
  ${\sigma_i=0.03}$ & $T$ (sec) & $\|e\|_2$ & ${\sigma_i=0.04}$ & 
  $T$ (sec) & $\|e\|_2$ \\
  \cmidrule(r){1-3} \cmidrule(l){4-6}
  \emph{Algorithm S$1$} & $6.6$ & $0.031$ & \emph{Algorithm S$1$} & 
  $8.9$ &$0.008$ \\
  \emph{Algorithm S$2$} & $6.6$ & $0.059$ & \emph{Algorithm S$2$} & 
  $8.8$ &0.073 \\
  \emph{Algorithm S$3$} & $6.6$ & $0.053$ & \emph{Algorithm S$3$} 
  & $8.8$ & $0.039$ \\
  \bottomrule
 \end{tabular}
 \caption{Comparison of algorithms for $p=196$ parameters.}
 \label{tab:compare_np196}
\end{table}

We see that algorithm S$1$ gives better approximation compared to 
the others as expected. Also the algorithm S$3$ performs 
significantly better compared to algorithm S$2$. Increasing the 
number of parameters gives better approximation as expected for 
algorithm $1$, though for the other algorithms this is not 
guaranteed due to the extra error incurred (see theorem 
\ref{thm:estimation_multiagent2}) which may increase with larger 
$p$ depending on other variables such as the location of centres.
$\sigma_i$ also plays an important role in the reconstruction of 
the original field. For $p=100$, $\sigma_i = 0.05$ seems to 
provide a better approximation compared to $\sigma_i = 0.04$, 
and for $p=196$, $\sigma_i = 0.04$ seems to provide a better 
approximation compared to $\sigma_i = 0.03$. To summarize, 
algorithm S$1$ gives better approximation compared to the others 
though it is more computational and memory intensive. The 
algorithm S$3$ also gives a good approximation requiring much less 
memory. It may also be noted that in many applications, we may 
only be interested in identifying the main features of the 
original field which was successfully done in most of the cases 
discussed. 

\section{Conclusion} \label{sec:conclusion}
In this paper we consider the estimation of a scalar field 
motivated by tools from adaptive control theory and lyapunov 
analysis. We derived two estimation algorithms, one in which 
each mobile sensor estimates the entire parameter vector, and 
another in which each mobile sensor estimates only part of the 
parameter vector. We verified and tested the algorithms using 
simulations. Further work involves improving upon the proposed 
algorithms, and possibility of estimation of time-varying fields 
by persistent motion of the mobile sensors. 

\ifCLASSOPTIONcaptionsoff
  \newpage
\fi
%
%



\bibliographystyle{IEEEtran}
\bibliography{IEEEabrv,mybibfile}
\end{document}